%% file: arxivDissILJournal.tex
\documentclass[10pt]{article}
\input{Packages} 										
\input{CustomCommands} 							

\begin{document}

\title{Dissipative Imitation Learning for Discrete Dynamic Output Feedback Control with Sparse Data Sets}
\author{Amy K. Strong, Ethan J. LoCicero, Leila J. Bridgeman}

\maketitle
\begin{abstract}
Imitation learning enables the synthesis of controllers for complex objectives and highly uncertain plant models. However, methods to provide stability guarantees to imitation learned controllers often rely on large amounts of data and/or known plant models. In this paper, we explore an \ac{io} stability approach to dissipative imitation learning, which achieves stability with sparse data sets and with little known about the plant model. A closed-loop stable dynamic output feedback controller is learned using expert data, a coarse \ac{io} plant model, and a new constraint to enforce dissipativity on the learned controller. While the learning objective is nonconvex, \ac{ico} and \ac{pgd} are explored as methods to successfully learn the controller. This new imitation learning method is applied to two unknown plants and compared to traditionally learned dynamic output feedback controller and neural network controller. With little knowledge of the plant model and a small data set, the dissipativity constrained learned controller achieves closed loop stability and successfully mimics the behavior of the expert controller, while other methods often fail to maintain stability and achieve good performance.
\end{abstract}
\section{Introduction}\label{sec1}

Complex desired behaviors and inaccurate plant models can make it difficult to formulate controller objectives. In these cases, imitation learning can be an effective alternative to traditional controller synthesis methods. Imitation learning bypasses direct controller design and instead uses expert demonstrations of desired system behavior to learn a system's reward function or directly learn a controller \cite{Osa2018}. Often no assumptions are made about the expert -- it may be some pre-designed controller or a human operator directly manipulating the system. This flexibility enables learning controllers that mimic behaviors difficult to directly characterize. However, unconstrained imitation learning does not provide stability guarantees. This paper proposes a new dissipativity-based approach to provide robust stability guarantees for a controller learned from \acf{io} expert data, using the coarse \ac{io} properties of a linear or nonlinear plant.

The simplest form of imitation learning, behavior cloning, is a supervised learning problem in which a mapping from state to action is learned through minimization of a loss function \cite{Osa2018}. While behavior cloning can create a policy that imitates a stable expert, there are no inherent stability properties of the learned policy. Once trained, if the learned controller encounters scenarios outside of the training data distribution, it may produce unhelpful, unexpected, or even catastrophic inputs, leading to compounding errors and instability \cite{Osa2018, codevilla2018end}. While issues of instability can be heuristically addressed through increasing the amount of training data or encouraging controller outputs that maintain the training data distribution \cite{kang2022lyapunov}, these solutions create a high data collection burden and do not give any theoretical stability guarantees. 

In situations that demand stability, robust control theory
is being applied to imitation learning \cite{Yin2022,Makdah2021,Palan2020,Havens2021, Pauli2021, Chen2018, Revay2020, Donti2021, east2022imitation, tu2022sample}. Recent efforts towards constrained imitation learning provide closed-loop stability or robustness guarantees for nominal linear time invariant (LTI)  and nonlinear plant models. Lyapunov stability theory was used with quadratic constraints to maximize the region of attraction of a closed loop system of an LTI plant and neural network feedback controller, while minimizing loss \cite{Yin2022}. Similarly, when fitting a policy to expert linear quadratic regulator (LQR) demonstrations, Lipschitz constraints were imposed on loss to ensure stability of feedback control \cite{Makdah2021}. Robust imitation learning has also been applied to linear feedback control policies through incorporation of prior knowledge about the expert demonstrator or the system itself. In Reference \cite{Palan2020}, a Kalman constraint was enforced on the learning process, which assumed that the expert demonstrator was LQR optimal. This work was later extended to some types of nonlinear plant models \cite{east2022imitation}. In Reference \cite{Havens2021}, robust stability was enforced by imposing a threshold on the closed-loop $\text{H}_{\infty}$ norm of a linear plant model during learning. In both References \cite{Palan2020} and \cite{Havens2021}, a stable policy was learned with a small amount of expert data. While model uncertainty was addressed in Reference \cite{Havens2021}, it remains unclear how to select an appropriate closed-loop $\text{H}_{\infty}$ norm if the plant model remains highly uncertain, which is an important use case of learning-based control. In fact, all these approaches are limited to cases where accurate models are available, and while alternative approaches that do not require a known plant model have been proposed \cite{sultangazin2021exploiting}, an assumed structure of the plant is still required. In contrast, \ac{io} stability theory relies only on coarse knowledge of the plant for stability guarantees - circumventing the need for accurate state-space models or restrictive assumptions.

This paper builds upon References \cite{Palan2020} and \cite{Havens2021}, but moves beyond the Small Gain Theorem, which underpins $\text{H}_{\infty}$ control, to use more general \ac{io} stability criteria to ensure stability of learned controllers. This is expected to be most helpful when the time-invariant plant model is unreliable and data is sparse. Using an \ac{io} perspective, the plant and controller can be modeled as mappings from inputs to outputs, and certain open-loop \ac{io} properties can be used to infer closed-loop stability. Importantly, these \ac{io} properties can often be shown from first principles to hold for nonlinear, time varying, and uncertain parameters, circumventing the problem of unreliable models altogether. Consider the Passivity Theorem \cite{Vidyasagar1977}, which shows that two strictly passive systems in negative feedback are stable. Many nonlinear physical systems are known to be passive for any set of parameters, so even if a passive plant is not well-modeled, it must be stabilized by a passive controller \cite{Geromel1997,Brogliato2007}. This approach to stability analysis has been generalized to passivity indices \cite{Vidyasagar1977}, conic sectors \cite{Zames1966}, dissipativity \cite{Hill1977}, and further \cite{Megretski1997,Safonov1980}. While these analysis and stabilization techniques have most often been applied to continuous-time systems, this paper offers a discrete-time approach more applicable to control design via data driven methods.

Designing optimal controllers that are constrained to satisfy desirable \ac{io} properties for robust stability is an active area of research. In particular, References \cite{Geromel1997} and \cite{Forbes2019} explore $\Hcal_2$-passive designs, Reference \cite{Sivaranjani2018} applied $\Hcal_2$-conic design to power system stabilization, and References \cite{Xia2020} and \cite{Scorletti2001} developed $\Hcal_\infty$-conic and $\Hcal_\infty$-dissipative designs, respectively. In these designs, an \ac{io} property is imposed on the controller as a linear matrix inequality (LMI) constraint during the performance optimization. However, none of these methods have been applied to robust imitation learning. 

Here, the dissipative imitation learning problem is posed to learn a discrete linear dynamic output feedback controller using only \ac{io} data and initial state information from a discrete-time plant while enforcing a desired QSR-dissipativity property -- extending results from Reference \cite{Strong2022} to learn the entirety of a discrete dynamic output feedback controller from \ac{io} data rather than relying on state estimates. The QSR-dissipativity framework encompasses passivity, passivity indices, bounded gain, and conic sectors \cite{Vidyasagar1981}, which allows for diverse applications of the control learning process. The behavior cloning objective is combined with an LMI constraint on the controller that enforces a desired QSR-dissipativity property. However, the proposed objective function is nonconvex. \Acf{ico} and \acf{pgd} are explored as methods to learn the controller, all while guaranteeing closed-loop stability using open-loop, coarse \ac{io} plant analysis. These methods are tested in simulation to learn expert controllers for both linear and nonlinear plant models.

\section{Preliminaries}
\subsection{Notation}
For a square matrix, $\bP \succ0$ denotes positive definite. Related properties (negative definiteness and positive/negative semi-definiteness) are denoted likewise. The identity matrix, zero matrix, and trace are $\I$, $\zero$, and $\tr(\cdot)$. Duplicate blocks in symmetric matrices are denoted $(*)$.  The quadruple $(\A,\B,\C,\D)$ denotes the LTI state space of a discrete-time system, $\x_{k+1} = \A\x_K + \B\bu_k$, $\y_k = \C\x_k + \D\bu_k$, with states $\x \in \mathbb{R}^n$, inputs $\bu \in \mathbb{R}^m$, and outputs $\y \in \mathbb{R}^p$. The system is defined on $k \in \mathbb{N}_+$.

Let $\Hcal$ be any Hilbert space with an inner product, $\langle f,\, g \rangle$, that is linear in its terms and with an induced norm $||\cdot|| = \sqrt{\langle f,\, g \rangle};$ this includes the $\Ell_2$ Lebesgue space and the $\ell_2$ sequence space. The $\ell_2$ and Frobenius norms are denoted $||\cdot||_2$ and $||\cdot||_F$. Recall, $\y \in \ell_2$ if $||\y||_2^2 = \langle \y,\, \y \rangle = \sum_{k=0}^\infty \y^T(k)\y(k) <\infty$, $k \in \mathbb{N}_+.$ Furthermore, $\y\in \ell_{2e}$, the extended $\ell_2$ space, if its truncation to $k\in[0,\,T]$ is in $\ell_2$ $\forall$ $T\geq 0$. Similarly, $\y \in \Ell_2$ if $||\y||_2^2 = \langle \y,\, \y \rangle = \int_0^\infty \y^T(t)\y(t)dt <\infty$, and $\y\in \Ell_{2e}$ if its truncation to $t\in[0,\,T]$ is in $\Ell_2$ $\forall$ $T\geq 0$. Extended Hilbert spaces are referred to generally as $\Hcal_e.$ Superscripts will be added to specify the dimension of sequences when needed to avoid ambiguity.

\subsection{Dissipativity and Special Cases} \label{sec:review}
A dissipative system is one that cannot produce energy; its stored energy is instead bounded from below and above by the its maximum extractable energy and the required supply \cite{Willems1972}. A special case of dissipativity, QSR-dissipativity, established a quadratic storage function for individual and interconnected control affine state-space systems, which resulted in relatively straightforward analysis for the Lyapunov stability of individual and interconnected QSR-dissipative systems \cite{Hill1976, Hill1977}. Historically, more focus has been designated to continuous-time dynamical systems. However, similar definitions for dissipativiy and QSR-dissipativity have been established for discrete-time control affine state space systems \cite{Byrnes1994, haddad2004vector}.

Vidyasagar \cite{Vidyasagar1981} constructed a broader definition of QSR-dissipativity as an \ac{io} property of an operator that maps between two Hilbert spaces. While originally defined specifically for an operator mapping between two $\Ell_2$ spaces, a more general definition is offered below. 
\begin{definition}\textit{(QSR-dissipativity \cite{Vidyasagar1981})}\label{def:qsr_def}
The operator $\Gcal: \Hcal_e^m \rightarrow \Hcal_e^p$ is dissipative with respect to the supply rate $w(\bu, \y) = \langle \y,\, \Q\y \rangle_T +\langle \y,\, \bS\bu \rangle_T + \langle \bu,\, \R\bu \rangle_T,$ if for all $\bu \in \Hcal_e^n$ and all $T > 0,$ $w(\bu, \Gcal\bu) \geq \beta,$ where $\beta$ is defined by the initial conditions.
\end{definition} 
The \ac{io} operator perspective avoids the need for a state space representation of dynamical systems. Moreover, there is no longer a need to distinguish between dissipativity and \ac{io} stability (defined below)  for continuous or discrete time systems, as both can be defined as mappings between Hilbert spaces, $\Ell_2$ or $\ell_2$ respectively \cite{carrasco2013towards}. 
\begin{definition}
    \textit{($\Hcal$ Input-Output (IO) Stability \cite{Vidyasagar1981})} An operator $\Gcal:\Hcal_e^n\rightarrow\Hcal_e^m$ is input-output stable if any input $\bu\in\Hcal^n$ maps to an output $\y\in\Hcal^m$.
\end{definition}
The \ac{io} stability of a QSR-dissipative operator can be established through the characteristics of the system's $\Q$ matrix. 
\begin{lemma}\textit{(QSR Input-Output Stability \cite{Vidyasagar1981})}
Suppose $\Gcal: \Hcal_e^m\rightarrow\Hcal_e^p$ is QSR dissipative. If $\Q$ is negative definite, then $\Gcal$ is \ac{io} stable.
\end{lemma}
\begin{proof}
    While Lemma 3.2.10 in Reference \cite{Vidyasagar1981} exclusively proves this result in the $\Ell_{2e}$ space, the proof requires an inner product space with an inner product linear in its terms, as the proof relies on the additive properties of integrals and summations. Therefore, any Hilbert space using its induced norm qualifies.
\end{proof}
Beyond the open loop properties of the system, QSR-dissipativity can be used to establish the \ac{io} stability of a closed-loop systems as well. The Network QSR Lemma \cite{Vidyasagar1981} provides a method to establish global QSR properties and stability for any number of QSR-dissipative subsystems. The following theorem will state only the requirements for two operators in negative feedback, as it is most relevant to this paper.
\begin{theorem}\label{thm:QSR}\textit{(QSR Input-Output Stability Theorem \cite{Vidyasagar1981}\cite{Hill1977}).} 
Consider two operators $\Gcal_i:\Hcal_e^m\rightarrow\Hcal_e^p$ that are $(\Q_i,\bS_i,\R_i)$-dissipative for $i=1,2$. Let their negative feedback interconnection be defined as $\bu_1=\br_1-\y_2$ and $\bu_2=\br_2+\y_1$. Then the closed loop from $\br^T=[\br_1^T,\,\br_2^T]$ to $\y^T = [\y_1^T,\,\y_2^T]$ is $\Ell_{2}$ stable if there exists $\alpha>0$ such that 
\begin{equation*}\label{eqn:QSRc}
    \bmat{\Q_1+\alpha\R_2 & -\bS_1+\alpha\bS_2^T \\ * & \R_1 + \alpha\Q_2}  \prec 0.
\end{equation*}
\end{theorem}
\begin{proof}
    Theorem 12 \cite{Vidyasagar1981} is proved exclusively in the $\Ell_{2e}$ space. However, the same requirements as stated in the proof of Lemma \ref{thm:QSR} enable an extension into any Hilbert space using its induced norm.
\end{proof}

QSR-dissipativity has many special cases including conic sectors \cite{Zames1966,Bridgeman2016}, passivity \cite{Brogliato2007}, and bounded gain \cite{Desoer1975}. Their relations to QSR-dissipativity are defined in \autoref{tbl:cases}. Using Theorem \ref{thm:QSR}, the \ac{io} stability of a negative feedback system of operators with specific QSR-dissipativity characteristics can be established.

\begin{table}
\centering
    \begin{tabular}{|c | c | c | c|}
        \hline
        Case & QSR & LMI \\
        \hline
        Passive & $\Q=\zero$, $\bS=\frac{1}{2}\I$, $\R=\zero$ & (\ref{eqn:KYPpass}) or (\ref{eqn:KYPdiss}) \\ 
        \hline
        $\gamma$-Bounded Gain & $\Q=-\I$, $\bS=\zero$, $\R=\gamma^2\I$ & (\ref{eqn:QSRc}) \\
        \hline
        Nondegenerate & &  \\ Interior Conic & $\Q=-\I$, $\R=-ab\I$ & \cite{Bridgeman2014}  \\ $a<0<b$ &$\bS=\frac{a+b}{2}\I$  & \\
        \hline
        Degenerate & &  \\ Interior Conic & $\Q=\zero$, $\R=-d\I$ & \cite{Bridgeman2014}  \\ $d<0$ &$\bS=\frac{1}{2}\I$  &  \\
        \hline
        QSR-dissipative & $\Q=\zero$, $\R \succeq0$, any $\bS$ & (\ref{eqn:QSRc}) or (\ref{eqn:KYPdiss})\\
        & $\Q \prec0$, $\R \succeq 0$ any $\bS$ &   \\
        \hline
    \end{tabular}
    \caption{A summary of important cases of QSR-dissipativity, their formulations, and the LMI(s) used to impose the property on an LTI system. This updates Table 1\cite{Hill1977} to include degenerate conic systems.}
    \label{tbl:cases}
\end{table}

While QSR properties of a system can be established and used for stability analysis, the design of systems that have desired QSR properties is also possible using \acp{lmi}. For discrete LTI systems, the discrete KYP Lemma provides conditions to test or design a system that satisfies passivity or QSR-dissipativity characteristics.
\begin{lemma}\label{lem:KYPpass}
\textit{(Discrete KYP Lemma \cite{Hitz})} A square stable LTI system $\Gcal:\ell_{2e}\rightarrow\ell_{2e}$ with minimal state space realization $(\A,\B,\C,\D)$ is passive if and only if there exists $\bP  \succ 0$ such that
\begin{equation}\label{eqn:KYPpass}
    \bmat{\A^T \bP \A - \bP & \A^T\bP\B - \C^T \\
    \B^T\bP \A - \C & \B^T\bP\B - (\D^T + \D)}  \preceq 0.
\end{equation}
    
\end{lemma}

\begin{lemma}\label{lem:KYPdiss}
    \textit{(Discrete Dissipativity Lemma \cite{Kottenstette2010})} A square stable LTI system $\Gcal:\ell_{2e}\rightarrow\ell_{2e}$ with minimal state space realization $(\A,\B,\C,\D)$ is QSR-dissipative if and only if there exists $\bP \succ 0$ such that 
    \begin{equation}\label{eqn:KYPdiss}
    \begin{bmatrix}
         \A^T\bP\A - \bP - \C^T\Q\C & \A^T\bP\B - (\C^T\bS + \C^T\Q\D) \\ (\A^T\bP\B - (\C^T\bS + \C^T\Q\D))^T & -\D^T\Q\D - \D^T\bS -\bS^T\D -\R  + \B^T\bP\B
    \end{bmatrix}  \preceq 0,
\end{equation}
\end{lemma}

\subsection{Connecting Dissipativity and Learning-Based Control}

When a plant is poorly understood due to parametric uncertainty, unmodeled nonlinearities, delays, etc, closed-loop conditions will not yield reliable stability guarantees. In this case, open-loop conditions based on coarse knowledge of plant \ac{io} properties (often derived from first principles) can be used to achieve robust stability guarantees without reliance on accurate state-space models. For example, it is well known that many systems -- such as RLC circuits, PID controllers, and flexible robotic manipulators -- are passive for all possible parameters \cite{Brogliato2007}. This fact follows from physical laws even for nonlinear and time-varying cases. By the Passivity Theorem \cite{Vidyasagar1977}, any passive controller must stabilize such systems. Thus, if an open-loop passivity constraint is imposed on the controller, closed-loop stability is guaranteed without resorting to a deficient LTI model. Moreover, while training and test data distributions may vary \cite{Osa2018}, passivity ensures closed-loop stability \cite{Brogliato2007}. This stability is guaranteed despite sparse training data or additional noise or inconsistencies in training data, which is especially relevant if a human expert is mimicked.

At its core, behavior cloning can be thought of as a system identification problem. The dynamic equations of the expert controller must be learned from its \ac{io} data so that the learned controller can sufficiently mimic its behavior. However, system identification problems require sufficient information from the \ac{io} data collected to accurately model a system. For example, much of system's identification literature relies on Willem's Fundamental Lemma, which states that a single trajectory of \ac{io} data from a linear system can contain all possible trajectories if the input to the system is sufficiently excited \cite{willems2005note}. These conditions have been extended for certain cases of nonlinear systems as well \cite{Berberich2020b}. However, determining and implementing an input with enough excitation is a potentially difficult process -- particularly when modelling more complex experts. The burden of collecting large amounts of low-noise data that captures all relevant behaviors is high. By enforcing dissipativity as a design constraint for the learned controller, closed loop stability can be guaranteed -- even in cases where there is a lack of salient data to successfully ``identify" the expert. The next section formalizes the dissipativity-constrained behavior cloning problem.

\section{Problem Statement}

Consider a plant $\Gcal:\Hcal_e^m \rightarrow \Hcal_e^p$ where $\bu,\,\y\in\ell_{2e}$ are inputs and outputs , $\mathbb{R}^m$ and $\mathbb{R}^p$, respectively. Consider also the LTI control law for dynamic output feedback $\mathcal{C}:(\Ahat,\Bhat,\Chat,\Dhat)$ with states $\xhat\in\mathbb{R}^{\nhat}$, inputs $\uhat\in\mathbb{R}^p$, and outputs $\yhat\in\mathbb{R}^m$. Let the controller and plant be in negative feedback defined by $\bu=\rhat-\yhat$ and $\uhat = \br +\y$, where $\rhat\in\mathbb{R}^m$ and $\br\in\mathbb{R}^p$ are noise.

 Suppose that an expert policy demonstration is defined by a sequence of plant-output/control-action pairs, $\{\y_k,\bu_k\}_{k=0}^N$, as well as the initial state of the plant, $\x_0$. The objective is to design a dynamic output feedback control law so that the controller, $\mathcal{C}$, closely imitates the behavior of the expert by producing the correct output, $\bu$, given some input, $\y$. The controller must also satisfy a prescribed QSR-dissipativity condition. This condition, in turn, ensures closed-loop stability through an associated \ac{io} stability theorem, like \autoref{thm:QSR}, even if the expert is imitated poorly or chosen poorly and would itself fail to ensure stability. 

In summary, dissipativity-constrained behavior cloning minimizes $\frac{1}{N}\sum_{k=0}^N \loss(\yhat,\bu_k)$ over $\mathcal{C}:(\Ahat,\Bhat,\Chat,\Dhat)$ such that $\mathcal{C}$ is $(\Q,\bS,\R)$-dissipative, where $\loss$ is a loss function that empirically measures how well the learned controller mimics the expert policy on the plant \ac{io} data. A simple and effective choice of loss is the \ac{mse}. Lemma \ref{lem:KYPdiss} can be applied to convert dissipativity condition into a matrix inequality constraint that is enforced on the controller as it is learned. The complete learning process can be described as
\begin{subequations} \label{eqn:nonconvex_problem}
    \begin{align}
        &\min_{\bP \succ0,\Ahat, \Bhat, \Chat, \Dhat} \optSpace \frac{1}{N}\sum_{k=0}^N ||\bu_k-\yhat_k||_2^2  \\
        &\text{s.t.} \optSpace \begin{bmatrix}
         \Ahat^T\bP\Ahat - \bP - \Chat^T\Q\Chat & \Ahat^T\bP\Bhat - (\Chat^T\bS + \Chat^T\Q\Dhat) \\ (\Ahat^T\bP\Bhat - (\Chat^T\bS + \Chat^T\Q\Dhat))^T & -\Dhat^T\Q\Dhat - \Dhat^T\bS -\bS^T\Dhat -\R  + \Bhat^T\bP\Bhat
    \end{bmatrix}  \preceq 0. \label{eqn:nonconvex_constraint}
    \end{align}
\end{subequations}
Constraint~\ref{eqn:nonconvex_constraint} is, in general, nonlinear. However, this constraint can be re-posed as an \ac{lmi}, as seen in the next section. While previous work \cite{Strong2022} was able to learn $\Chat$ of $\mathcal{C}$ using a convex objective with convex constraints, learning all elements of $\mathcal{C}$ requires a non-convex objective function. \ac{ico} and  \ac{pgd} are explored to address the problem of learning $\mathcal{C}$ despite the non-convex optimization problem.

\section{Main Results}

A new LMI constraint is proposed for imposing QSR-dissipativity with $\mathcal{C}: (\Ahat, \Bhat, \Chat, \Dhat)$ and $\bP$ as the design variables. QSR-dissipativity includes passivity, as well as more specific design criteria such as bounded gain and interior conic bounds. These special cases are tabulated in terms of their equivalent QSR-dissipativity characterization in \autoref{tbl:cases}. Together, the results of Lemma~\ref{lem:KYPdiss} and Corollary~\ref{lem:QSRc} are combined to provide a unified framework for imposing any interior conic bounds, and more generally any QSR property with $\Q \prec0$, when designing the dynamic output feedback controller, $\mathcal{C}$.
\begin{corollary}\label{lem:QSRc}
Let the LTI system $\Gcal:(\A,\B,\C,\D)$ be controllable and observable. Then $\Gcal$ is $(\Q,\bS,\R)$-dissipative with $\Q  \prec 0$ if and only if there exists $\bP  \succ 0$ such that 
    \begin{equation}\label[cons]{eqn:almost_convex_constraint}
        \begin{bmatrix}
            -\bP & -\C^T\bS & \C^T & \A^T \\
            * & -\D^T\bS - \bS^T\D - \R & \D^T & \B^T \\
            * & * & \Q^{-1} & \mathbf{0} \\
            * & * & * & -\bP^{-1}
        \end{bmatrix}  \preceq 0.
    \end{equation}
\end{corollary}
\begin{proof}
\autoref{eqn:nonconvex_constraint} is represented as 
\begin{equation}
     \begin{bmatrix}
        - \bP & - (\Chat^T\bS) \\ *  &  - \Dhat^T\bS -\bS^T\Dhat -\R 
    \end{bmatrix} - \X^T \Q \X + \Y^T \bP \Y  \preceq 0, 
\end{equation}
where $\X = \bmat{\Chat & \Dhat}$ and $\Y = \bmat{\Ahat & \Bhat}.$ A Schur complement \cite{boyd1994linear} is performed about terms $\X^T\Q\X$ and $\Y^T\bP\Y.$ Note that the matrix, $\Q$, must be negative definite.
\end{proof}

\begin{remark}\label{ref: rem1} While \autoref{eqn:almost_convex_constraint} results in a nonlinear matrix inequality, the constraint is convexified through an overbound on $\bP^{-1}.$ The inverse of $\bP$ is bounded above by $2\tilde{\bP}^{-1} - \tilde{\bP}^{-1}\bP\tilde{\bP}^{-1}$ \cite{CaverlyArXiv}, where $\tilde{\bP}$ is some positive definite matrix, such as a previous iteration of $\bP.$ Then, \autoref{eqn:almost_convex_constraint} negative semi-definite if and only if there exists $\bP  \succ 0$ such that 
    \begin{equation}\label[cons]{eqn:convex_constraint}
        \begin{bmatrix}
            -\bP & -\C^T\bS & \C^T & \A^T \\
            * & -\D^T\bS - \bS^T\D - \R & \D^T & \B^T \\
            * & * & \Q^{-1} & \mathbf{0} \\
            * & * & * & -2\tilde{\bP}^{-1} +\tilde{\bP}^{1}\bP\tilde{\bP}^{-1}
        \end{bmatrix}  \preceq 0.
    \end{equation}
By substituting Equation \ref{eqn:convex_constraint} for Equation \ref{eqn:nonconvex_constraint}, the learning problem in Equation \ref{eqn:nonconvex_problem} now has convex constraints. 
\end{remark}

However, the objective function of \autoref{eqn:nonconvex_problem} remains non-convex in terms of $\Ahat, \Bhat, \Chat,$ and $\Dhat$. The following sections show two approaches, \ac{ico} and \ac{pgd}, to formulating and solving Problem~\ref{eqn:nonconvex_problem}, while imposing \cref{eqn:convex_constraint}. These methods are then evaluated in terms of their efficacy when mimicking experts for both linear and non-linear plants.

\subsection{Iterative Convex Overbounding}

The initial controller state ($\xhat_0$) and known controller input ($\uhat$) can be propagated forward using controller dynamics, $\mathcal{C}: (\Ahat, \Bhat, \Chat, \Dhat)$, to predict controller output ($\yhat$). This approach is similar to system identification literature, where large matrices known as the observability and Toeplitz matrix, are constructed and used to learned the linear dynamics of a plant based on its inputs and outputs \cite{van1994n4sid}. However, our method uses only two time steps and assumes a known initial state. Using \ac{mse} as the loss function, the imitation learning problem becomes
\begin{subequations}\label[prob]{eq:min_forward_prop_prob}
\begin{align}\label{eq:min_forward_prop}
    \min_{\Ahat,\Bhat,\Chat,\Dhat} \frac{1}{N}\sum_{k=0}^N \loss(\bu_k, \yhat_k) = \min_{\Ahat,\Bhat,\Chat,\Dhat}\frac{1}{N}\sum_{k=0}^N \left |\begin{bmatrix}
         \bu_{k,0} \\\bu_{k,1}
    \end{bmatrix} - \begin{bmatrix}
         \Chat\xhat_{k,0} + \Dhat\uhat_{k,0} \\ \Chat(\Ahat\xhat_{k,0} + \Bhat\uhat_{k,0}) + \Dhat\uhat_{k,1}
    \end{bmatrix}\right |_2^2
    \\
    \text{s.t.      } \text{\cref{eqn:convex_constraint}},
\end{align}
    
\end{subequations}
where $\bu$ is the true controller output and $N$ is the number of training data trajectories. 

\autoref{eq:min_forward_prop} involves only two time steps of forward propagation, but the objective quickly becomes non-convex, with polynomial terms resulting from the $\ell_2$ norm. Rather than optimizing a non-convex objective function, \ac{ico} can be used to convexify the objective function through overbounding techniques \cite{Warner2017a, sebe2018sequential, de2000convexifying}. \ac{ico} is an overbounding technique that is typically implemented when facing bilinear matrix inequality constraints that result from some $\mathcal{H}_2$ and $\mathcal{H}_\infty$ design problems \cite{Warner2017a}, but it can easily be extended to nonconvex objective functions. While \ac{ico} may not converge to the global optimum, convergence to a local optimum is guaranteed. Moreover, \ac{ico} scales well to high dimensions, and at the point of overbounding, the objective function has no conservatism \cite{Warner2017a}. In this paper, a series of \acp{lmi} are developed to overbound the polynomial terms in \autoref{eq:min_forward_prop}-- creating a new objective function that is then minimized in place of the true \ac{mse}. 

In \ac{ico}, a design variable, $\V$, is broken down into a constant, $\V_0$, and a design variable, $\delta \V.$ The constant $\V_0$ is a feasible solution to the overbound constraints, while $\delta \V$ is the design variable used to optimize the objective function. The following lemmas describe overbounding techniques used for non-convex terms found in the objective function. Lemmas \ref{lem:ico_2} and \ref{lem:ico_3} are specific instances of the general polynomial bound developed by Reference \cite{Warner2017a}. The following lemmas can be used to change an expression that is polynomial in a matrix to an equivalent expression that is linear, facilitating numerical solutions.

\begin{lemma}\label{lem:ico_1}
For $\M=\M^\trans\in \mathbb{R}^{n \times n}$, $\V\in\mathbb{R}^{n\times p}$, $\Z\in\mathbb{R}^{p\times n}$, and $m,n,p\in\mathbb{Z}$, $\Z^T\V^T\V \Z \prec \M$ is equivalent to
\begin{equation}\label{eq:lmi_ico1}
    \bmat{- \M + \Gamma(\V) & \Z^T\delta \V^T \\
    * & -\mathbf{I}}  \prec 0,
\end{equation}
where $\V = \V_0 + \delta \V$, and $\Gamma(\V) = \Z^T\V_0^T\delta\V \Z + \Z^T\delta\V^T\V_0 \Z + \Z^T\V_0^T\V_0\Z.$ 
\end{lemma}
\begin{proof}
     Bound $\Z^T\V^T\V \Z$ from above with $\M$ so that $\Z^T\V^T\V \Z - \M \prec 0.$ Expand $\Z^T\V^T\V \Z$ using $\V = \delta\V + \V_0$ to get $\Gamma(\V) +\Z^T\delta\V^T \mathbf{I}\delta\V \Z - \M \prec0$. Perform a Schur complement about $\Z^T\delta\V^T \mathbf{I}\delta\V \Z.$
\end{proof}

Note that Lemma \ref{lem:ico_1} is not needed to overbound the quadratic elements of the objective function. However, higher order nonlinearities within the objective function require this lemma to create \acp{lmi}.
    
\begin{lemma}\label{lem:ico_2}
For $\M=\M^\trans\in \mathbb{R}^{n \times n}$, $\bS\in \mathbb{R}^{m \times n}$, $\U\in \mathbb{R}^{m \times n}$, $\V\in\mathbb{R}^{n\times p}$, $\W=\W^\trans\in \mathbb{R}^{n \times n}$, $\Z\in\mathbb{R}^{p\times n}$, and $m,n,p\in\mathbb{Z}$, $\text{He}(\bS^T\U\V \Z) \prec \M$ if and only if
\begin{equation}\label{eq:lmi_ico2}
    \begin{bmatrix}
         - \M + \Gamma(\U, \V) & \bS^T\delta \U & \Z^T\delta \V^T \\ * & -\W^{-1} & 0 \\
         * & *  & -\W
    \end{bmatrix}  \prec 0,
\end{equation}
where $\Gamma(\U, \V) = \text{He}(\bS^T\delta\U\V_0 \Z + \bS^T\U_0\delta\V \Z + \bS^T\U_0\V_0\Z)$, $\U = \U_0 + \delta \U,$ and $\V = \V_0 + \delta \V$.
\end{lemma}
\begin{proof}
The Hermition $\text{He}(\Z^T\V^T\U \bS)$ is expanded using $\V = \V_0 + \delta\V$ and $\bU = \bU_0 + \delta \bU$ to get $\Gamma(\U,\V) + \text{He}(\bS^T\delta\U\delta\V\Z).$ Through Young's relation \cite{zhou1988robust}, $\text{He}(\bS^T\delta\U\delta\V\Z) \leq \bS^T\delta\U^T\W\delta\U \bS + \Z^T \delta\V^T \W^{-1} \delta\V \Z,$ where $\W$ is a symmetric, positive definite matrix. Bound the expression $\bS^T\delta\U^T\W\delta\U \bS + \Z^T \delta\V^T \W^{-1} \delta\V \Z + \Gamma(\U, \V)$ from above using matrix design variable, $\M$, and then perform a Schur complement about $\bS^T\delta\U^T\W\delta\U \bS$ and $\Z^T \delta\V^T \W^{-1} \delta\V \Z.$
\end{proof}

\begin{lemma}\label{lem:ico_3}
For $\M=\M^\trans\in \mathbb{R}^{n \times n}$, $\V\in\mathbb{R}^{n\times p}$, $\Z\in\mathbb{R}^{p\times n}$, $\U \in \mathbb{R}^{m \times n}$, $\T \in \mathbb{R}^{m \times m},$ and $m,n,p\in\mathbb{Z}$, the cubic Hermition expression $\text{He}(\Z^T\V^T\U^T\T\bS) \prec \M$ if and only if
\begin{equation}\label{eq:lmi_ico3}
     \bmat{\bmat{- \M +\Gamma(\V, \U, \T) & \Z^T(\V_0^T\delta\U^T + \delta\V^T\U_0^T) & \bS^T\delta\T^T \\ * & -\W_1^{-1} & \mathbf{0} \\ * & * & -\W_1 } & \bmat{\Z^T \\ \mathbf{0} \\ \mathbf{0}}\delta\V^T & \bmat{\bS^T\T_0^T \\ \I &\\\mathbf{0}}\delta\U \\ * & \W_2^{-1} & \mathbf{0} \\
    * & * & -\W_2}  \prec 0,
\end{equation}
where $\Gamma(\V,\U,\T) = \He(\Z^T(\delta\V^T\U_0^T\T_0 + \V_0^T\delta\U^T\T_0 +\V_0^T\U_0^T\delta\T +\V_0^T\U_0^T\T_0)\bS),$ $\T = \T_0 + \delta\T, \V = \V_0 + \delta\V_0, \U = \U_0 + \delta \U,$ and $\W_1$ and $\W_2$ are some symmetric, positive definite matrices.
\end{lemma}
\begin{proof}
Let $\U\V = \F,$ $\T = \T_0 + \delta\T$, and $\F = \F_0 + \delta \F.$ Using Lemma \ref{lem:ico_2}, $\text{He}(\Z^T\F^T\T \bS)$ is overbounded with the \ac{lmi}
\[ \bmat{-\M +\He(\Z^T(\F^T\T_0 + \F_0^T\delta\T)\bS) & \Z^T\delta\F^T & \bS^T\delta\T^T \\ * & -\W^{-1} & \mathbf{0} \\ * & * & -\W }  \prec 0.\]
This expression can be expanded using $\V\U = \F$ to get
\[\bmat{- \M +\Gamma(\U,\V,\T) & \Z^T(\V_0^T\delta\U^T + \delta\V^T\U_0^T) & \bS^T\delta\T^T \\ (\U_0\delta\V + \delta\U\V_0)\Z & -\W^{-1} & \mathbf{0} \\ \delta\T\bS & \mathbf{0} & -\W } + 
 \He\left(\bmat{\Z^T \\ \mathbf{0} \\ \mathbf{0}}\delta\V^T\delta\U^T\bmat{\T_0\bS & \I & \mathbf{0}} \right)  \prec 0.\]
 Using Young's relation with a Schur complement once again, results in \autoref{eq:lmi_ico3}.
\end{proof}


\begin{lemma}\label{lem:ico_4}
For $\M=\M^\trans\in \mathbb{R}^{n \times n}$, $\V\in\mathbb{R}^{n\times p}$, $\U \in \mathbb{R}^{m \times n}$, and $m,n,p\in\mathbb{Z}$, the quartic matrix expression $\z^T\V^T\U^T\U\V \z \prec \M$ if and only if
   \begin{equation}\label{eq:lmi_ico4}
    \bmat{\bmat{-\M + \Gamma(\V,\U) & \z^T(\delta \V^T \U_0^T + \V_0^T\delta\U^T) \\ * & -\I} & \bmat{\z^T\V_0^T\U_0^T \\ \I}\delta \U & \bmat{\z^T \\ \mathbf{0}} \delta \V \\ * & -\W^{-1} & \mathbf{0} \\ * & * & -\W } \prec0,
   \end{equation}
where $\Gamma(\U,\V) = \z^T(\V_0^T\U_0^T\U_0\V_0 + \V_0^T\U_0^T\delta\U \V_0 + \V_0^T\U_0^T\U_0 \delta \V  + \delta \V^T\U_0^T\U_0 \V_0 + \V_0^T\delta \U^T \U_0 \V_0)\z$, $\V = \V_0 + \delta\V, \U = \U_0 + \delta\U$, and $\W$ is some symmetric, positive definite matrix. 
\end{lemma}
\begin{proof}
Let $\F = \U\V$ and $\F = \F_0 + \delta \F$. Using Lemma \ref{lem:ico_1}, $\z^T\F^T\F\z$ is overbounded with the \ac{lmi}
\[\bmat{-\M + \z^T(\F_0^T\F_0 + \F_0^T\delta\F +\delta\F^T\F_0)\z & x^T\delta\F^T \\ \delta\F & -\I}  \prec 0.\]
Now recall that $\F_0 = \U_0\V_0$ and $\delta\F = \U_0\delta \V_0 + \delta \U \V_0 + \delta \U \delta \V.$ The previous \ac{lmi} can be expanded to
\[\bmat{-\M + \Gamma & \z^T(\delta \V^T \U_0^T + \V^T\delta\U^T) \\ (\delta\U \V_0 + \U_0 \delta \V)\z & -\I} +\He\left(\bmat{\z^T\V_0^T\U_0^T \\ \I}\delta \U \delta \V \bmat{\z & \mathbf{0}}\right)  \prec 0.\] 
Once again, use Young's relation and Schur complement to acheive \autoref{eq:lmi_ico4}.
\end{proof}

When overbounding \cref{eq:min_forward_prop_prob}, the matrix design variables $\mathcal{C}:(\Ahat, \Bhat, \Chat,\Dhat)$ become $\mathcal{C} + \delta \mathcal{C}:$ $(\Ahat_0 + \delta \Ahat, \Bhat_0 + \delta \Bhat, \Chat_0 + \delta \Chat, \Dhat_0 + \delta \Dhat)$.
The objective function in \cref{eq:min_forward_prop_prob} is expanded to
\begin{gather}\label[prob]{eq:min_forward_prop_exp}
\begin{split}
    \min_{\delta\Ahat, \delta\Bhat, \delta\Chat, \delta\Dhat} \frac{1}{N}\sum_{k=0}^N \bu_{k,0}^T\bu_{k,0} - \He(\bu_{k,0}^T\Chat\xhat_{k,0}) - \He(\bu_{k,0}^T\Dhat \uhat_{k,0}) + \He(\xhat_{k,0}^T\Chat^T\Dhat\uhat_{k,0}) + 
    \\
    \xhat_{k,0}^T\Chat^T\Chat\xhat_{k,0} + \uhat_{k,0}^T\Dhat^T\Dhat\uhat_{k,0} +
    \bu_{k,1}^T\bu_{k,1}^T - \He(\bu_{k,1}^T\Chat\Ahat\xhat_{k,0}) -\He(\bu_{k,1}^T\Chat\Bhat\uhat_{k,0}) \\ -\He(\bu_{k,1}^T\Dhat\uhat_{k,1}) + \xhat_{k,0}^T\Ahat^T\Chat^T\Chat\Ahat\xhat_{k,0} + 
    \He(\xhat_{k,0}^T\Ahat^T\Chat^T\Dhat\uhat_{k,1}) + \\ 
    \He(\uhat_{k,0}^T\Bhat^T\Chat^T\Chat\Ahat\xhat_{k,0}) + 
    \uhat_{k,0}^T\Bhat^T\Chat^T\Chat\Bhat\uhat_{k,0} + \He(\uhat_{k,0}^T\Bhat^T\Chat^T\Dhat\uhat_{k,1}) + \uhat_{k,1}^T\Dhat^T\Dhat\uhat_{k,1}.
\end{split}
\end{gather}
Using Lemmas 4-7, each non-convex term within the objective function is bounded using a design variable $m_{k,i}.$  These design variables are then substituted into \autoref{eq:min_forward_prop} to create a conservative, convex objective, while a corresponding \ac{lmi} is added to the constraints of \cref{eq:min_forward_prop_prob} ($\text{LMI}_{k,i}$) to preserve the term's overbound. This process creates the optimization problem,
\begin{subequations}\label[prob]{eq:min_forward_prop_cvx}
    \begin{align}
        \begin{split}\label{eq:obj_ico}
            \min_{\delta\Ahat, \delta\Bhat, \delta\Chat, \delta\Dhat} \frac{1}{N}\sum_{k=0}^N \bu_{k,0}^T\bu_{k,0} - \He(\bu_{k,0}^T\Chat\xhat_{k,0}) - \He(\bu_{k,0}^T\Dhat \uhat_{k,0}) + \He(\xhat_{k,0}^T\Chat^T\Dhat\uhat_{k,0}) \\ + \xhat_{k,0}^T\Chat^T\Chat\xhat_{k,0} + \uhat_{k,0}^T\Dhat^T\Dhat\uhat_{k,0} +
            \bu_{k,1}^T\bu_{k,1}^T - m_{k,1} -m_{k,2} \\ -\He(\bu_{k,1}^T\Dhat\uhat_{k,1}) + m_{k,3} + m_{k,4} + 
            m_{k,5} + m_{k,6} + m_{k,7} + \uhat_{k,1}^T\Dhat^T\Dhat\uhat_{k,1}
        \end{split}
\\
    &\quad \quad \quad \quad \quad \quad\quad \quad \quad \quad \quad \quad\quad \quad \quad \quad \quad \quad\quad \quad \quad \quad \quad \quad\quad \quad \quad \quad\quad \text{s.t.      }\text{\cref{eqn:convex_constraint}}
\\  \label[cons]{eq:overbound_constr}
      &\quad \quad \quad \quad \quad \quad\quad \quad \quad \quad \quad \quad\quad \quad \quad \quad \quad \quad\quad \quad \quad \quad \quad \quad\quad \quad \quad \quad\quad\quad\quad\text{LMI}_{k,1:7} \prec 0.
    \end{align}
\end{subequations}
Overbounding the objective function results in numerous \ac{lmi} constraints -- dependent on the total amount of data, as each data point requires overbounding for each nonconvex element of the objective function with which it interacts. These \acp{lmi} are represented as $\text{LMI}_{k,1:7}$ in \cref{eq:min_forward_prop_cvx}. 

To learn a controller using the new, convex objective described in \cref{eq:min_forward_prop_cvx}, an initial feasible controller, $\mathcal{C}_0$, is required. 
A random initialization scheme can be used to encourage exploration of the solution space. In the random initialization, $\mathcal{C}_0:(\Ahat_0, \Bhat_0, \Chat_0, \Dhat_0)$, is generated by sampling each matrix element from a normal distribution. Then, $\mathcal{C}_0$ is projected into the QSR-dissipative space, using \autoref{eqn:projection}.
Once a feasible initial controller is found, \cref{eq:min_forward_prop_prob} is iteratively solved, using $\delta\mathcal{C}$ to update the next iteration's starting point. This process is repeated until the change in objective function from one iteration to the next is minimal. \cref{alg:ico} summarizes the process of imitation learning with \ac{ico}.

\begin{algorithm}
\caption{ICO for Dissipative Imitation Learning}\label[alg]{alg:ico}
    \begin{algorithmic}
        \State \textbf{Input: }Training Data: $(\bu, \y, \x_0)$
        \State \textbf{Init: }$\mathcal{C}_0: (\Ahat_0, \Bhat_0, \Chat_0, \Dhat_0)$ using Random Initialization
        \State \textbf{Set: }$i = 0$, $\Delta \loss = \infty$
        \While{$\Delta \loss(\bu, \yhat) \geq \beta$}
        \State \textbf{Set: }$\mathcal{C}_0 = \mathcal{C}_i$
        \State \textbf{Solve: } \cref{eq:min_forward_prop_cvx}
        \State $\mathcal{C}_{i+1} = \mathcal{C}_i + \delta \mathcal{C}$
        \State $\Delta \loss = |\loss_{i} - \loss_{i-1}|$
        \State $i = i+1$
        \EndWhile
        \State \textbf{Output: }Learned Controller: $\mathcal{C}_{i+1}: (\Ahat, \Bhat, \Chat, \Dhat)$
    \end{algorithmic}
\end{algorithm}
\begin{theorem}
    Assume that the controllable and observable LTI system, $\mathcal{C}_{i =0}:(\Ahat, \Bhat, \Chat, \Dhat),$ is an initial controller for \cref{alg:ico}. Let $\mathcal{C}_0$ be dissipative with respect to some $\Q,$ $\bS$, and $\R$ and satisfy \cref{eq:overbound_constr}. Then, $\mathcal{C}_{i+1} = \mathcal{C}_i + \delta \mathcal{C}$  is also $\Q\bS\R$ dissipative, where $\delta \mathcal{C}$ is the minimizer of \cref{eq:min_forward_prop_cvx}. Further, as \cref{alg:ico} iterates, the objective function of \autoref{eq:obj_ico}, $\loss$, is non-increasing, leading to convergence to a local minimum, $\mathcal{C}^*$.
    \begin{proof}
        Using Lemma \ref{lem:QSRc}, $\mathcal{C}_0$ satisfies \autoref{eqn:almost_convex_constraint}. By induction, every iteration of \ac{ico}, $\mathcal{C}_{i+1}: (\Ahat_i + \delta\Ahat, \Bhat_i + \delta\Bhat, \Chat_i + \delta\Chat, \Dhat_i + \delta\Dhat)$, is feasible and therefore $\mathcal{C}_{i}$ satisfies \autoref{eqn:almost_convex_constraint} for $i\geq0$. Consequently, these controllers are all dissipative. Moreover, $\delta\mathcal{C}$ either improves the loss of $\mathcal{C}_{i+1}$ or is zero, leading to a non-increasing loss function across iterations. Because the loss function is bounded below by zero, convergence to a local minima is guaranteed.
    \end{proof}
\end{theorem}

Some convex overbounding techniques, Lemma \ref{lem:ico_3} for example, require the use of a weighting matrix, $\W.$ This weighting matrix can either be pre-set or treated as an additional design variable. While its value does not affect the overall result of the overbounding process, it can affect the rate of convergence \cite{Warner2017a}. For this \ac{ico} process, the weighting matrices were treated as design variables -- constrained to be symmetric positive definite matrices. The non-convex term $\W^{-1}$  is overbounded by the inequality $\W^{-1}  \succeq 2\tilde{\W} - \tilde{\W}\W\tilde{\W}$, where $\tilde{\W}$ is a constant and can be initialized as some positive definite matrix.

\subsection{Projected Gradient Descent}


\ac{pgd} is a special case of proximal gradient descent and  does not require a convex objective function \cite{bertsekas1997nonlinear}. Like \ac{ico}, the non-convex objective function for \ac{pgd} uses controller dynamics, $\mathcal{C}: (\Ahat, \Bhat, \Chat, \Dhat)$, to propagate the initial controller state, $\xhat_0$  (assumed to be equivalent to the initial plant state) and the known series of controller inputs, $\uhat_k,$ forward in time to predict a series of outputs, $\yhat.$ For some number of time steps, $n$, $\yhat$ can be predicted using,
\begin{equation}\label{eq:proj_forward}
    \bmat{\yhat(0) \\\yhat(1) \\ \yhat(2) \\ \vdots \\ \yhat(n)} = \bmat{\Chat \\ \Chat\Ahat \\ \Chat\Ahat^2 \\ \vdots \\ \Chat\Ahat^{n}}\xhat(0) + \bmat{\Dhat & 0 & 0 & ... & 0 \\ \Chat\Bhat & \Dhat & 0 & ... & 0 \\\Chat\Ahat\Bhat & \Chat\Bhat & \Dhat & ... & 0 \\ \vdots & \vdots & \vdots & ... & \vdots \\ \Chat\Ahat^{n}\Bhat & \Chat\Ahat^{n-1}\Bhat & \Chat\Ahat^{n-2}\Bhat& ...& \Dhat}\bmat{\bu(0) \\ \uhat(1) \\ \uhat(2) \\ \vdots \\ \uhat(n)}.
\end{equation}
\autoref{eq:proj_forward} uses matrices commonly known as the observability matrix and Toeplitz matrix in systems identification literature \cite{van1994n4sid}. The predicted $\yhat$ is then compared to the true control inputs, $\bu$, using \ac{mse} loss,
\begin{equation}\label{eq:loss_pgd}
    \loss(\bu, \yhat) = \frac{1}{N}\sum_{i=1}^N||\bu - \yhat||_2^2.
\end{equation}

In gradient descent, the gradient of objective function, \autoref{eq:loss_pgd}, is scaled by the learning rate, $\epsilon$, and used to update the controller parameters -- producing $\tilde{\mathcal{C}}: (\Atilde, \Btilde, \Ctilde, \Dtilde).$ Note that $\tilde{\mathcal{C}}$ is not guaranteed to have the required dissipativity characteristics to ensure closed loop stability. After each update, the parameters of $\tilde{\mathcal{C}}$ must be projected into the desired $\Q\bS\R$-dissipative convex subspace to produce the controller, $\mathcal{C}$, which maintains desired dissipativity characteristics. This projection is defined as  
\begin{equation}\label{eqn:projection}
    \begin{aligned}
       \Pi(\tilde{\mathcal{C}}) = \argmin_{\Ahat, \Bhat, \Chat, \Dhat \text{ satisfying (\ref{eqn:convex_constraint})}} ||(\Atilde, \Btilde, \Ctilde, \Dtilde) - (\Ahat, \Bhat, \Chat, \Dhat)||_F^2,
    \end{aligned}
\end{equation}
which is similar to that of Reference \cite{Havens2021}. \cref{eqn:convex_constraint} creates a subspace of the design variables that is closed and convex. Therefore, the projection is guaranteed to be unique \cite{nesterov2003introductory}. \ac{pgd} for dissipative imitation learning is summarized in \cref{alg:pgd}. It can be efficiently implemented using python autograd libraries, such as Pytorch \cite{paszke2019pytorch}.

\begin{algorithm}
\caption{PGD for Dissipative Imitation Learning}\label[alg]{alg:pgd}
    \begin{algorithmic}
        \State \textbf{Input: }Training Data: $(\bu, \y, \x_0)$
        \State \textbf{Set: }$i = 0$, $\Delta \loss = \infty$
        \State \textbf{Init: }$\mathcal{C}_i: (\Ahat_i, \Bhat_i, \Chat_i, \Dhat_i)$ using Random Initialization
        \While{$\Delta \loss(\bu, \yhat) \geq \beta$}
        \State $\tilde{\mathcal{C}} = \mathcal{C}_i + \epsilon\nabla\loss(\bu, \yhat)|_{\mathcal{C}}$ 
        \State $\mathcal{C}_{i+1} = \Pi(\Atilde, \Btilde, \Ctilde, \Dtilde)$
        \State $\Delta \loss = |\loss_{i} - \loss_{i-1}|$
        \State $i = i+1$
        \EndWhile
        \State \textbf{Output: }Learned Controller: $\mathcal{C}_{i+1}: (\Ahat, \Bhat, \Chat, \Dhat)$
    \end{algorithmic}
\end{algorithm}

Some asymptotic convergence guarantees exist for \ac{pgd} when minimizing a nonconvex objective functions in a closed, convex subspace \cite{bertsekas1997nonlinear}. However, the objective function must be continuously differentiable and its gradient must be Lipschitz continuous. The step size must be carefully selected using the Lipschitz constant. This result has been extended to non-asymptotic cases as well \cite{ghadimi2016mini}. In practice, the Lipschitz constant can be difficult to determine. Instead, a diminishing step size can be chosen to increase the likelihood of descent \cite{bertsekas1997nonlinear}. An exponentially decaying learning rate was used for \cref{alg:pgd}. Although \ac{pgd} lacks convergence assurances, its ease of implementation and speed make it a useful candidate for training controllers.

\section{Numerical Examples}

The proposed behavior cloning procedures were used to learn expert behavior for control of a linear, passive system and a nonlinear, QSR-dissipative system. The expert controller of the linear system is an LQR-optimal state feedback controller, but the learned controller only has access to output data, representing the ``privileged expert'' training case. The expert of the nonlinear system is a more complex \ac{nmpc}, representing the case where learning is used to imitate a controller whose actions demand computation times that are untenable for real-time control. In both examples, two types of sparse training data sets are examined: one assumes access to only initial plant states and trains the controllers with small initial trajectory segments, while the other assumes constant access to plant states and uses a sliding window of trajectory segments to train the controller. These examples demonstrate the viability of a learned dynamic output feedback controller with sparse data sets and show the ability of the learned controllers to mimic complex desired behavior, while maintaining closed loop stability.

\subsection{Control of a Linear System}

Dynamic output feedback controllers were learned for mass-spring-damper systems with a chain of three unit masses. Ten randomized, passive mass-spring-damper systems were created with spring and damping coefficients randomly sampled from uniform distributions over $[1, 5]$ and $[1, 10]$ respectively. The inputs to the system are a force applied to each unit mass, while the system outputs velocity with force feedthrough. The feedthrough matrix was designed as $\phi \I$, where $\phi$ is a scalar randomly sampled from the uniform distribution over $[0.05, 1].$ 

For each of the ten randomized systems, an expert policy was designed as a noisy LQR-optimal state-feedback controller $\bu=-\K\x + \e$, where $\e$ is noise, $\x$ is the state of the system, and $\K=\E^{-1}\B^T\Pi$, where $\Pi$ solves $\A^T\Pi + \Pi\A -\Pi\B\E^{-1}\B^T\Pi + \C^T\F\C = \zero,$ and $\E$ and $\F$ are $10$ and $1.$ Expert data to train the learned controllers was generated by randomly sampling an initial state $\x_0$ from a Gaussian distribution, $N(0,1)$ and then controlling the system states to a reference of zero while inputs ($\bu$), outputs ($\y$), and states ($\x$) were recorded. 

Each controller was learned using either \ac{ico} or \ac{pgd} using \ac{lmi} \ref{eqn:convex_constraint} as the method of enforcing passivity, with $\Q = 0,$ $\bS = \frac{1}{2}\I$, and $\R = 0.$ Because \ac{ico} uses only two data points per trajectory, the same standard was enforced on all controllers in the training process. Both methods of using training data, initial segments and a sliding window, were used for training and were compared.
Both \ac{pgd} and \ac{ico} were trained for 150 iterations, or until the change in the cost function between two iterations was less than $|1e-3|.$ Dissipativity was enforced on the learned controller throughout the training process. To provide comparison, dynamic output feedback controllers were learned using unconstrained \ac{ico} (ICO-NC) and standard \ac{gd} -- each of which have no inherent stability guarantees. An additional comparison was made by training a \ac{nn} controller with \ac{gd} (NN+GD) to imitate the expert data. The network was three layers -- each layer with 150 hidden units -- with the \ac{relu} activation function acting on each layer of the network.

The final performance of each learned controller was evaluated by comparing the states of the controlled system when controlled by the expert versus when controlled by the learned controller. New trajectories created from initial conditions drawn randomly from the distribution $N(0,20^2)$ were used to test the controllers. \ac{mse} is used as the evaluation metric, $\frac{1}{N}||\x_e(k) - \x(k)||_2^2$, where $\x_e(k)$ are the plant states when controlled by the expert, $\x(k)$ are the plant states when controlled by the learned controller, and $N$ is the number of time steps in the simulation. Plant and controller noise added in simulation are $N(0,0.04^2)$ and $N(0,0.02^2)$ respectively. 

\autoref{tab:linear_multTraj} and \autoref{tab:linear_Segments} show results from training and testing each of the learned controllers. While \ac{ico} and \ac{pgd} learned controllers remained stable 100\% of the time during testing, both \ac{ico}-NC and \ac{gd} learned controllers rarely, if ever, remained stable. Of the controllers with no stability guarantees, the learned \ac{nn} controllers had the highest percentage of controllers remaining stable. Note also that the \ac{pgd} and \ac{ico} training processes often did not result in the most change in cost, which may indicate that the other learning processes overfit to the training data and therefore performed poorly in testing.

\autoref{fig:linear_median} shows the median \ac{mse} of the learned controllers for both training data types, while the poorest controller performance is shown in \autoref{fig:linear_max} and the quartiles of controller performance in \autoref{fig:linear_quartiles}. The performance of learned controllers trained using \ac{ico} without dissipativity constraints or \ac{gd} have been excluded from all figures, as their performance was so poor, the \ac{mse} was greater than $1^{100}.$ \ac{pgd} and \ac{ico} learned controllers performed similarly well on the test data set, with \ac{ico} controllers performing slightly better with additional data. However, the \ac{ico} training process required significantly more training time in comparison to \ac{pgd}.

   
\begin{table}[h]
\centering
\hspace*{-2.2\leftmargin}
\begin{tabular}{ |p{1.5cm}|| p{0.5cm}|p{0.5cm}|p{0.5cm}|p{0.5cm}|p{0.5cm}||p{0.95cm}|p{0.95cm}|p{0.95cm}|p{0.95cm}|p{0.95cm}||p{0.79cm}|p{0.79cm}|p{0.79cm}|p{0.79cm}|p{0.79cm}||  }
 \hline
 \multicolumn{16}{|c|}{Learned Controller - Initial Segments} \\
 \hline
 \multirow{3}{*}{Training} & \multicolumn{5}{|c|}{Percent Stable} & \multicolumn{5}{|c|}{Percent Change in Cost} & \multicolumn{5}{|c|}{Training Time (s)}\\
 \cline{2-16}
  & \multicolumn{15}{|c|}{Number of Segments} \\
 \cline{2-16}
  Method & 2 & 25 & 50 & 100 & 200 & 2 & 25 & 50 & 100 & 200 & 2 & 25 & 50 & 100 & 200 \\
  \hline
  ICO & 100 & 100 & 100 & 100 & 100 & -93.22 & -84.32 & -82.04 & -82.42 & -82.08 & 28.40 & 3.22e2 & 6.80e2 & 1.34e3 & 2.76e3\\
  \hline
  ICO-NC& 24 & 0 & 0 & 0 & 0 & -99.28 & -99.58 & -99.63 & -99.73 & -99.76 & 68.02 & 7.34e2& 1.49e3&2.88e3 &6.34e3\\
  \hline
  PGD & 100 & 100 & 100 & 100 & 100 & -75.40 & -46.58 & -39.50 & -38.37 & -37.79 & 14.32 & 14.88 & 15.73 & 15.78 & 15.77 \\
  \hline
  GD& 0 & 0 & 0 & 0 & 1 & -95.35 & -88.77 & -87.77 & -87.68 & -87.70 & 2.47 & 2.77 &  2.80 & 2.86 & 2.99 \\
  \hline
  GD+NN& 21 & 31 & 20 & 20 & 20 & -95.32 & -50.11 & -49.89 & -46.41 & -45.14 & 0.27 & 0.34 & 0.42 & 0.63 & 0.79\\
  \hline

\end{tabular}

\caption{Percentage of controllers that remain stable during testing, percent change in cost during training, and average training time when using initial segments as training data trajectories.}
\label{tab:linear_multTraj}
\end{table}
\begin{table}[h]
\centering
\hspace*{-2.2\leftmargin}
\begin{tabular}{|p{1.5cm}|| p{0.5cm}|p{0.5cm}|p{0.5cm}|p{0.5cm}|p{0.5cm}||p{0.95cm}|p{0.95cm}|p{0.95cm}|p{0.95cm}|p{0.95cm}||p{0.79cm}|p{0.79cm}|p{0.79cm}|p{0.79cm}|p{0.79cm}|| }
 \hline
 \multicolumn{16}{|c|}{Learned Controller - Sliding Window} \\
 \hline
 \multirow{3}{*}{Training} & \multicolumn{5}{|c|}{Percent Stable} & \multicolumn{5}{|c|}{Percent Change in Cost} & \multicolumn{5}{|c|}{Training Time (s)}\\
 \cline{2-16}
  & \multicolumn{15}{|c|}{Number of Segments} \\
 \cline{2-16}
  Method & 2 & 25 & 50 & 100 & 200 & 2 & 25 & 50 & 100 & 200 & 2 & 25 & 50 & 100 & 200\\
  \hline
  ICO & 100 & 100 & 100 & 100 & 100 & -98.16 & -95.75 & -94.32 & -92.79 & -92.13 & 25.70 & 97.28 & 1.49e2& 3.65e2&8.60e2 \\
  \hline
  ICO-NC&43& 10& 5& 2& 0& -99.05 & -98.97 & -98.67 & -98.78 & -99.12 & 60.46 &2.55e2 & 3.89e2& 8.34e2&2.13e3\\
  \hline
  PGD& 100 & 100 & 100 & 100 & 100 & -85.23 & -73.04 & -65.85 &-64.01 & -66.65 & 14.51 & 11.93 & 11.17 & 11.35 & 11.90 \\
  \hline
  GD& 1 & 5 & 5 & 5 & 5 & -96.02 & -90.38 & -87.60 & -87.61 & -88.93 & 2.42 & 2.19 & 2.12 & 2.25 & 2.43 \\
  \hline
  GD+NN& 65 & 41 & 39 & 26 & 50 & -95.14 & -84.85 & -77.51 & -65.61 & -53.95 & 0.25 & 0.47 & 0.54 & 0.66 & 1.02\\
  \hline

\end{tabular}

\caption{Percentage of controllers that remain stable during testing, percent change in cost during training, and average training time when using a sliding window as training data trajectories.}
\label{tab:linear_Segments}
\end{table}

\begin{figure}[htp]
    \centering
     \hspace*{\fill}
    \begin{subfigure}[t]{0.45\textwidth}
        \centering
        \resizebox{\textwidth}{!}{\input{Figures/linear_multTraj_median}}
        \caption{Controllers were trained using initial segments of trajectories.}
        \label{subfig:multTraj_linear_median}
    \end{subfigure}
    \hfill
     \begin{subfigure}[t]{0.45\textwidth}
        \centering
        \resizebox{\textwidth}{!}{\input{Figures/linear_segments_median}}
        \caption{Controllers were trained using a sliding window of trajectories.}
        \label{subfig:segments_linear_median}
    \end{subfigure}
    \hspace*{\fill}
    \caption{The median \ac{mse} of plant states when controlled with learned controllers compared to the expert controller.}
    \label{fig:linear_median}
\end{figure}
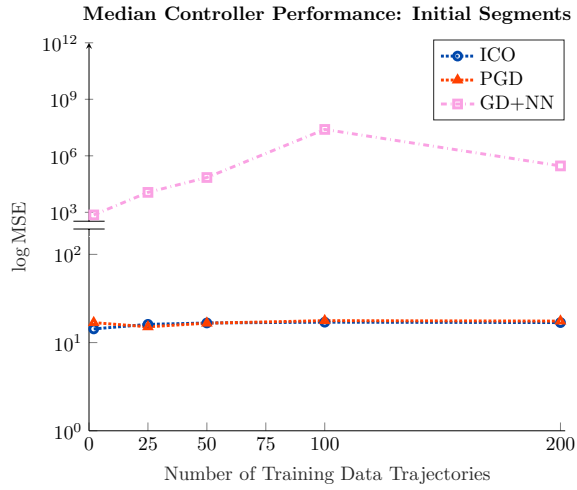
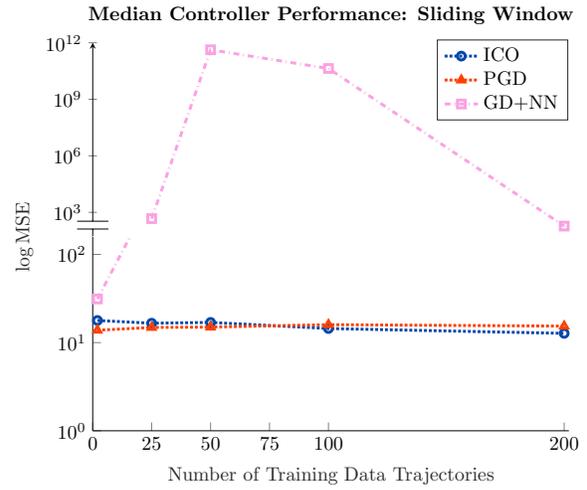

\vspace{-100pt}

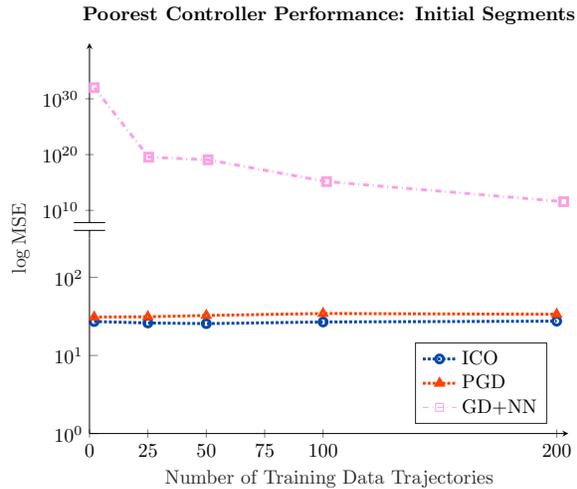
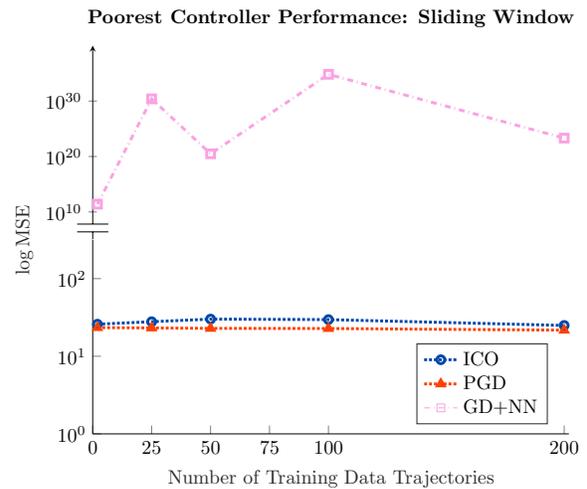
\begin{figure}[htp]
    \centering
    \hspace*{\fill}
    \begin{subfigure}[t]{0.45\textwidth}
        \centering
        \resizebox{\textwidth}{!}{\input{Figures/linear_multTraj_max}}
        \caption{Controllers were trained using initial segments of trajectories.}
        \label{subfig:multTraj_linear_max}
    \end{subfigure}
    \hfill
     \begin{subfigure}[t]{0.45\textwidth}
        \centering
        \resizebox{\textwidth}{!}{\input{Figures/linear_segments_max}}
        \caption{Controllers were trained using a sliding window of trajectories.}
        \label{subfig:segments_linear_max}
    \end{subfigure}
    \hspace*{\fill}
    \caption{The highest \ac{mse} of plant states when controlled with learned controllers compared to the expert controller.}
    \label{fig:linear_max}
\end{figure}

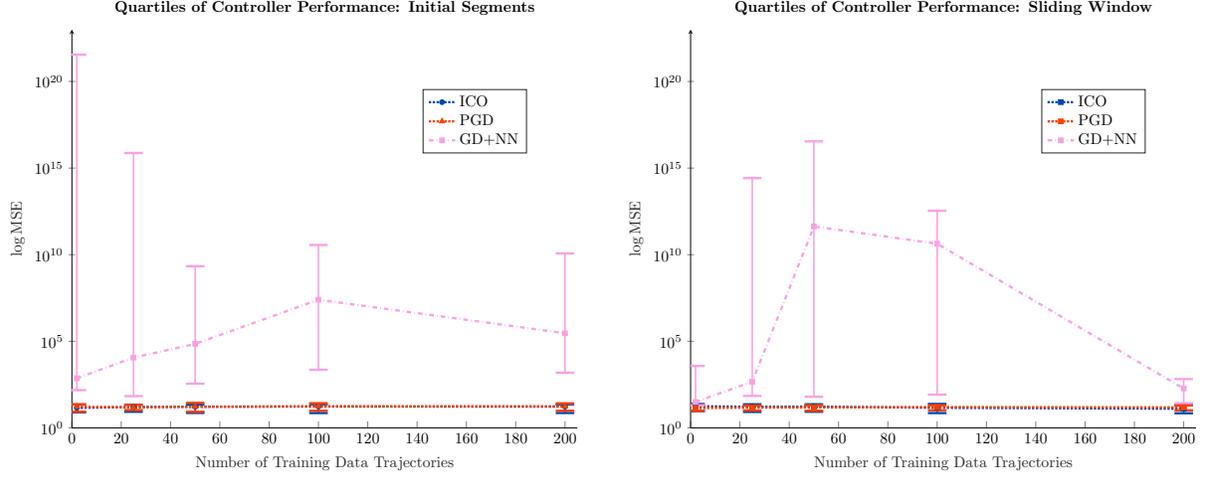
\begin{figure}[htp]
    \centering
    \hspace*{\fill}
    \begin{subfigure}[t]{0.45\textwidth}
        \centering
        \resizebox{\textwidth}{!}{\input{Figures/linear_multTraj_quartiles}}
        \caption{Controllers were trained using initial segments of trajectories.}
        \label{subfig:multTraj_linear_quartiles}
    \end{subfigure}
    \hfill
     \begin{subfigure}[t]{0.45\textwidth}
        \centering
        \resizebox{\textwidth}{!}{\input{Figures/linear_segments_quartiles}}
        \caption{Controllers were trained using a sliding window of trajectories.}
        \label{subfig:segments_linear_quartiles}
    \end{subfigure}
    \hspace*{\fill}
    \caption{Quartiles of the \ac{mse} of plant states when controlled with learned controllers compared to the expert controller.}
    \label{fig:linear_quartiles}
\end{figure}

\pagebreak
\subsection{Control of a Nonlinear System}

The imitation learning scheme was also tested on a QSR-dissipative discrete, nonlinear plant controlled by an expert \ac{nmpc}. The plant is described as
\begin{equation}\label{eq:nonlinSys}
    \Gcal_p : \begin{cases}
    \x_{k+1} = \A \x_k + \alpha\sin(\x_k) + \B \bu_k \\
    \y_k = \C \x_k  + \D \bu_k,
    \end{cases}
\end{equation}
where $x \in \mathbb{R}$, $u \in \mathbb{R}$, $y \in \mathbb{R}$, $(\A,\B,\C,\D)$ is defined as $(0.5, 1, 1, -2),$ and $\alpha$ is $0.4$. The \ac{nmpc} controls the system to equilibrium through minimizing the cost function $\sum_{i=1}^{N-1} \x_i^T\Q\x_i + \bu_i^T\R\bu_i + \x_N^T\Q\x_N^T,$ with $N = 11,$ $\Q = 2$, and $\R = 1.$ Training data trajectories were generated by simulating the expert controlling the system to equilibrium for randomized initial conditions, drawn from the distribution $N(0, 5^2)$. The initial state, as well as plant output, controller input, and plant states were recorded for each trajectory. Plant noise of  $N(0,0.04^2)$ and input noise of $N(0,0.02^2)$ were added in simulation.

The QSR-dissipativity of the nonlinear plant, $\Gcal_p$, was determined using the following \ac{lmi}.
\begin{lemma}
The dynamical system, $\Gcal_p$, is QSR dissipative if there exists $\bP  \succ 0$ such that
    \begin{equation}\label{eq:nonlinQSR}
        \bmat{\C^T\Q\C - 2\A\bP\A + \bP - 2\alpha^2\bP &  \C^T\Q\D +\C^T\bS - 2\A^T\bP\B \\ \D^T\Q\C + \bS^T\C-2\B^T\bP\A  & \D^T\Q\D + \bS^T\D + \D^T\bS + \R - 2\B^T\bP\B}  \succeq 0.
    \end{equation}

\begin{proof}

    From Definition \ref{def:qsr_def}, $\Gcal_p$ is QSR dissipative if the supply rate $w(\bu,\y) = \langle \y,\Q\y \rangle + \langle \y, \bS\bu \rangle + \langle \bu, \R\bu\rangle \geq \beta,$ where $\beta$ depends on initial conditions. Define the candidate Lyapunov function $V = \x_k^T \bP \x_k$ and its difference operator
    \begin{equation}
         \Delta V_k = \x_{k+1}^T\bP\x_{k+1} - \x_k^T\bP\x_k = (\A\x_k + \alpha\sin( \x_k) + \B\bu_k)^T\bP(\A\x_k + \alpha\sin(\x_k) + \B\bu_k) - \x_k^T\bP\x_k,
    \end{equation}
    where $\bP$ is a positive definite matrix. Using a telescoping summation, it can be shown that 
    \begin{equation}
        V(\x_T) - V(\x_0) = \sum_{i=0}^T\Delta V_i.
    \end{equation}
    
    S-procedure is used to find the QSR characteristics of the system. Define $\beta = V(\x_0),$ $S_T = \sum_{i=0}^T\Delta V_i + V(\x_0),$ and $L_T = \sum_{i=0}^T \y_i^T\Q\y_i + 2\y_i^T\bS\bu_i + \bu_i\R\bu_i.$ Note that because $S_T$ is positive definite, if $L_T - S_T > \beta,$ then $L_T > \beta.$ The inequality,  $L_T - S_T > \beta,$ becomes
    \begin{equation}\label{eq:s_pro}
        \sum_{i=0}^T \y_i^T\Q\y_i + 2\y_i^T\bS\bu_i + \bu_i\R\bu_i - \sum_{i=0}^T\Delta V_i \geq 0,
    \end{equation}
    which can also be written as
    \begin{multline}
     \sum_{k = 0}^{T-1} \bmat{\x_k \\ \bu_k}^T\bmat{\C^T\Q\C - \A^T\bP\A - \bP & \C^T\Q\D + \C^T\bS - \A^T\bP\B \\ \D^T\Q\C + \bS^T\C - \B^T\bP\A & \D^T\Q\D + \bS^T\D + \D^T\bS + \R -\B^T\bP\B}\bmat{\x_k \\ \bu_k} \\ - \alpha \sin(\x_k)^T\bP (\alpha \sin\x_k) - (\A\x_k +\B\bu_k)^T\bP \alpha \sin(\x_k) - \alpha\sin^T(\x_k)\bP(\A\x_k + \B\bu_k)  \geq 0.
    \end{multline}
    
    To overbound the $\alpha\sin( \x_k)^T\bP \alpha\sin(\alpha \x_k)$, note that $\alpha\sin( \x_k)^T\bP \alpha\sin( \x_k) \leq \alpha^2\x_k^T\x_k$ and therefore 
    
    \noindent$-\alpha\sin( \x_k)^T\bP \alpha\sin( \x_k) \geq -\alpha^2\x_k^T\bP \x_k.$ Furthermore, through completion of the square 
    \[ - \alpha \sin(\x_k)^T\bP \alpha \sin(\x_k) - (\A\x_k +\B\bu_k)^T\bP \alpha \sin(\x_k) - \alpha\sin^T(\x_k)\bP(\A\x_k + \B\bu_k)  \geq \]\[-\alpha^2\x_k^T\x_k - (\A\x_k + \B\bu_k)^T\bP(\A\x_k + \B\bu_k)\]
    Using these two overbounds, \autoref{eq:s_pro} is reformulated as the \ac{lmi} \autoref{eq:nonlinQSR}.
\end{proof}
\end{lemma}

QSR-dissipativity of $\Gcal_p$ was found to hold for $\Q_p = -1$, $\bS_p = 3.556$, and $\R_p = 29.333,$ where $\Q,$ $\bS$, and $\R$ were found through minimizing the conic radius of the plant \cite{joshi2002design}. Using \autoref{thm:QSR} with $\alpha = 1$, a stabilizing controller is required to be QSR-dissipative with respect to $\Q_c= -29.867$, $\bS_c= 3.556$, and $\R_c=0.601$. The learned controller was trained using expert \ac{nmpc} trajectories, while dissipativity with respect to $\Q_c\bS_c\R_c$ was enforced through \cref{eqn:convex_constraint}. 

A dynamic output feedback controller was learned using \ac{ico} with and without dissipativity constraints, using \ac{pgd} with dissipativity constraints, and \ac{gd} without any dissipativity constraints. In each instance, the dynamic output feedback controller had the same number of states, inputs, and outputs as the plant. A \ac{nn} controller was also trained using \ac{gd} as a point of comparison for the learned controllers. The \ac{nn} was a three layer network, with 25 hidden units per layer and \ac{relu} activation functions. These hyperparameters were selected because they resulted in generally good controller performance. Note that the neural network requires learning seventy-five network weights $(\mathbf{W}_{1},\mathbf{W}_{2},\mathbf{W}_{3})$, as opposed to the dynamic output feedback controller's four system parameters, $(\Ahat, \Bhat, \Chat, \Dhat).$

Each learned controller was trained using 100 randomized parameter initializations, $(\Ahat_0, \Bhat_0, \Chat_0, \Dhat_0)$ (or for the \ac{nn} controller - $(\mathbf{W}_{0,1},\mathbf{W}_{0,2},\mathbf{W}_{0,3})$). For each controller, training was terminated when the difference in cost between two iterations was less than $|1e-5|.$  All learned controllers were trained using either only initial segments of data, representing a scenario in which little is known about the plant state, and a sliding window of trajectory segments on a small number of training data trajectories. Because the \ac{nmpc} controlled the system to equilibrium quickly, only the first ten steps of the training data trajectories were used for the sliding window approach. The learned controllers were then evaluated by comparing their performance to that of the expert when controlling the plant to equilibrium for one thousand new initial conditions, drawn from the distribution $N(0, 15^2)$. The \ac{mse} between the plant states when controlled by the expert, $\x_e$, and the plant states when controlled by the learned controller, $\x$ was used the characterize the learned controller's ability to imitate the expert. 

\autoref{tab:nonlin_multTraj} and \autoref{tab:nonlin_Segments} show the percentage of learned controllers that remain stable when in closed loop with $\Gcal_p$, as well as the average percent change in cost function during training and the average total amount of training time needed for each controller. All of the \ac{ico} and \ac{pgd} trained controllers remained stable while controlling $\Gcal_p,$ as expected. Surprisingly, 100\% of the \ac{nn} controllers remained stable as well, but without any theoretic stability guarantees. While \ac{ico}, by far, had the longest training time, it also had the best performance, even with low amounts of data.

The median \ac{mse} of the learned controllers is displayed in \autoref{fig:nonlinear_median}, while the  poorest controller performance and the quartiles of \ac{mse} are shown in \autoref{fig:nonlinear_max} and \autoref{fig:nonlinear_quartiles}. Learned dynamic output feedback controllers trained with \ac{ico}-NC or \ac{gd} (without dissipativity constraints) performed incredibly poorly, often becoming unstable. While the \ac{nn} controller performed worse than \ac{ico} or \ac{pgd} trained dynamic output feedback controllers, it did remain stable throughout testing. Of the controllers that remained stable throughout testing, the \ac{pgd} learned controller had the greatest variation in performance. Overall, the \ac{ico} learned and dissipativity constrained controller performed best -- outperforming the other methods even when it was trained with little data.

\begin{table}[htp]
   \centering
\begin{tabular}{|p{2.25cm}|| p{0.5cm}|p{0.5cm}|p{0.5cm}|p{0.5cm}||p{0.92cm}|p{0.92cm}|p{0.92cm}|p{0.92cm}||p{0.7cm}|p{0.7cm}|p{0.7cm}|p{0.7cm}|| }
 \hline
 \multicolumn{13}{|c|}{Learned Controller - Initial Segments} \\
 \hline
 \multirow{3}{*}{Training} & \multicolumn{4}{|c|}{Percent Stable} & \multicolumn{4}{|c|}{Percent Change in Cost} & \multicolumn{4}{|c|}{Training Time (s)}\\
 \cline{2-13}
   & \multicolumn{12}{|c|}{Number of Segments} \\
 \cline{2-13}
  Method & 2 & 25 & 50 & 100 & 2 & 25 & 50 & 100 & 2 & 25 & 50 & 100 \\
  \hline
  ICO & 100 & 100 & 100 & 100 & -88.14 & -96.31 & -97.09 & -96.90 & 1.44 & 8.52 & 16.93 & 33.02 \\
  \hline
  ICO-NC & 72 & 63 & 63 & 63 & -98.54 & -99.37 & -99.53 & -99.46  & 3.06 & 12.74 & 25.44 & 49.81\\
  \hline
  PGD & 100 & 100 & 100 & 100 & -96.90 & -97.37 & -97.29 & -97.36 & 11.45 & 12.22 & 12.08 & 12.27  \\
  \hline
  GD& 32 & 32 & 32 &  32 & -28.85 & -28.51 &  -28.61 & -28.54 & 2.50 & 2.58 & 2.62 & 2.67  \\
  \hline
  GD+NN & 100 & 100 & 100 & 100 & -99.59 & -98.76 & -98.92 & -98.85 & 0.10 & 0.08 & 0.10 & 0.12 \\
  \hline

\end{tabular}

\caption{The percentage of controllers that remain stable when controlling from test initial conditions is recorded, as well as the average percent change in cost for each controller during training and the average training time. Controllers were trained with the initial segments of the training data trajectories.}
\label{tab:nonlin_multTraj}
\end{table}

\begin{table}[htp]
\centering
\begin{tabular}{|p{2.25cm}|| p{0.5cm}|p{0.5cm}|p{0.5cm}|p{0.5cm}||p{0.92cm}|p{0.92cm}|p{0.92cm}|p{0.92cm}||p{0.7cm}|p{0.7cm}|p{0.7cm}|p{0.7cm}|| }
 \hline
 \multicolumn{13}{|c|}{Learned Controller - Sliding Window} \\
 \hline
 \multirow{3}{*}{Training} & \multicolumn{4}{|c|}{Percent Stable} & \multicolumn{4}{|c|}{Percent Change in Cost} & \multicolumn{4}{|c|}{Training Time (s)}\\
 \cline{2-13}
   & \multicolumn{12}{|c|}{Number of Segments} \\
 \cline{2-13}
 Method & 2 & 25 & 50 & 100 & 2 & 25 & 50 & 100 & 2 & 25 & 50 & 100 \\
  \hline
  ICO & 100 & 100 & 100 & 100 & -97.75 & -94.98 & -95.06 & -95.36 & 13.93 & 106.71 & 43.13 & 81.02 \\
  \hline
  ICO-NC & 87 & 75 & 67 & 64 & -99.63 & -99.39 & -99.46 & -99.47 &18.60&135.51&180.25&330.79\\
  \hline
  PGD & 100 & 100 & 100 & 100 & -97.47 & -97.17 & -97.34 & -97.32 & 12.04 & 11.30 & 11.25 & 11.30 \\
  \hline
  GD & 32 & 32 & 32 & 32 & -28.67 & -27.61 & -28.03 & -27.78 & 2.39 & 2.25 & 2.37 & 2.38  \\
  \hline
  GD+NN & 100 & 100 & 100 & 100 & -98.78 & -99.43 & -99.14 & -98.98 & 0.09 & 0.09 & 0.10 & 0.12 \\
  \hline

\end{tabular}

\caption{The percentage of controllers that remain stable when controlling from test initial conditions is recorded, as well as the average percent change in cost for each controller during training and the average training time. Controllers were trained with a sliding window of segments from 0 to 20 training data trajectories.}
\label{tab:nonlin_Segments}
\end{table}

\section{Conclusion}
Covariate shift between the training and test data sets can cause catastrophic consequences for imitation learned controllers. In this paper, \ac{io} stability theory was used to develop and impose constraints on a dynamic output feedback controller during imitation learning to ensure closed loop stability with a given plant. This methodology guarantees stability even with sparse expert training data sets that likely do not comprehensively capture controller behavior. Moreover, the controller parameters were learned directly from \ac{io} data with a known initial state.

Two methods for learning the controller were explored -- \ac{ico} and \ac{pgd}. While both training methods synthesized controllers that performed well when controlling a linear system, \ac{ico} synthesized controllers with much better performance on nonlinear systems. Unfortunately, the computational burden and burden of implementation for \ac{ico} is much higher than that of \ac{pgd}. However, the development of toolboxes\cite{zhu2022development} may provide the opportunity for more accessible implementation of \ac{ico}.

\section*{Acknowledgments}
This material is based upon work supported by NSF GFRP Grant No. 1644868, the Alfred P. Sloan Foundation, and ONR YIP Grant No. N00014-23-1-2043.

\begin{figure}
    \centering
     \hspace*{\fill}
    \begin{subfigure}[t]{0.45\textwidth}
        \centering
        \resizebox{\textwidth}{!}{\input{Figures/nonlinear_multTraj_median}}
        \caption{Controllers were trained using initial segments of trajectories.}
        \label{subfig:multTraj_nonlinear_median}
    \end{subfigure}
    \hfill
     \begin{subfigure}[t]{0.45\textwidth}
        \centering
        \resizebox{\textwidth}{!}{\input{Figures/nonlinear_segments_median}}
        \caption{Controllers were trained using a sliding window of trajectories.}
        \label{subfig:segments_nonlinear_median}
    \end{subfigure}
    \hspace*{\fill}
    \caption{Median \ac{mse} performance of the learned controllers controlling a nonlinear when compared to the expert controller.}
    \label{fig:nonlinear_median}
\end{figure}
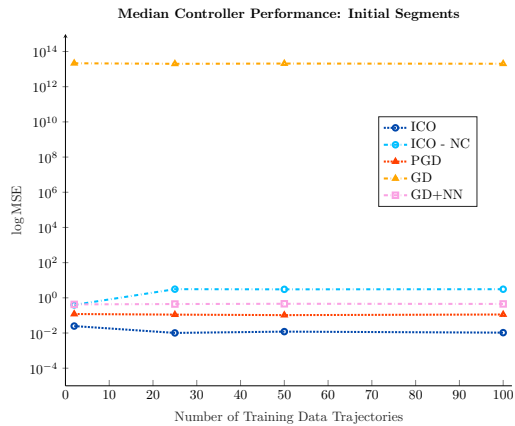
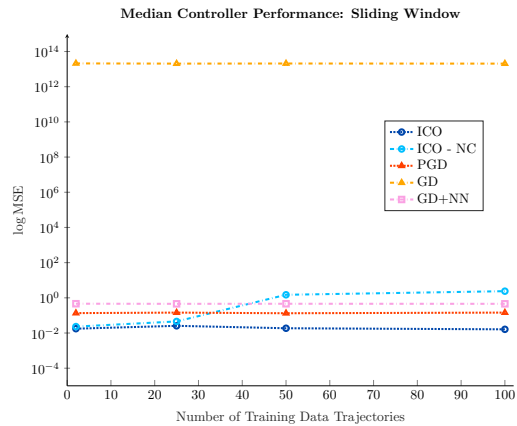

\begin{figure}
    \centering
    \hspace*{\fill}
    \begin{subfigure}[t]{0.45\textwidth}
        \centering
        \resizebox{\textwidth}{!}{\input{Figures/nonlinear_multTraj_max}}
        \caption{Controllers were trained using initial segments of trajectories.}
        \label{subfig:multTraj_nonlinear_max}
    \end{subfigure}
    \hfill
     \begin{subfigure}[t]{0.45\textwidth}
        \centering
        \resizebox{\textwidth}{!}{\input{Figures/nonlinear_segments_max}}
        \caption{Controllers were trained using a sliding window of trajectories.}
        \label{subfig:segments_nonlinear_max}
    \end{subfigure}
    \hspace*{\fill}
    \caption{Poorest \ac{mse} performance of the learned controllers controlling a nonlinear when compared to the expert controller.}
    \label{fig:nonlinear_max}
\end{figure}
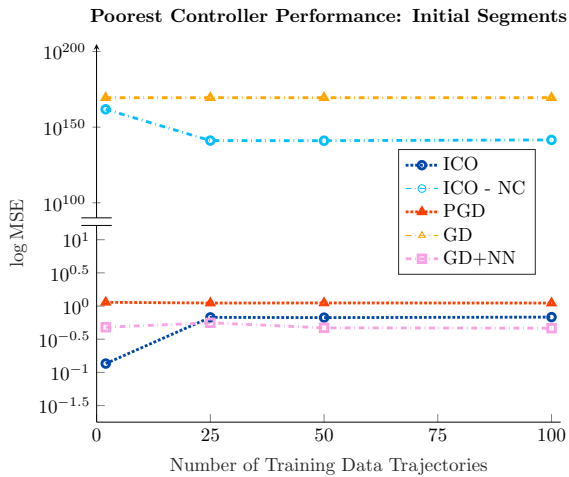
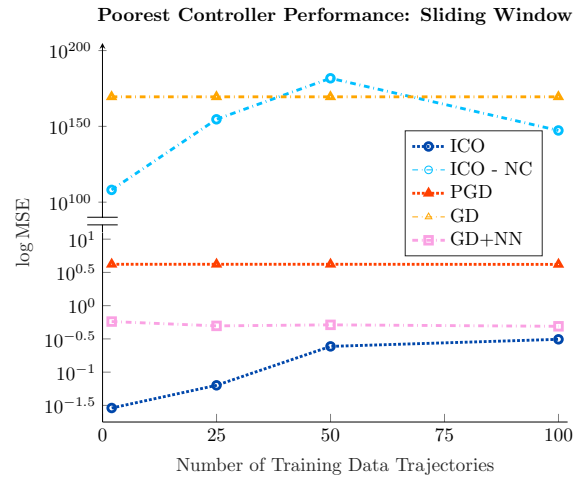

\begin{figure}
    \centering
     \hspace*{\fill}
    \begin{subfigure}[t]{0.4\textwidth}
        \centering
        \resizebox{\textwidth}{!}{\input{Figures/nonlinear_multTraj_quartiles}}
        \caption{Controllers were trained using initial segments of trajectories.}
        \label{subfig:multTraj_nonlinear_quartiles}
    \end{subfigure}
    \hfill
     \begin{subfigure}[t]{0.4\textwidth}
        \centering
        \resizebox{\textwidth}{!}{\input{Figures/nonlinear_segments_quartiles}}
        \caption{Controllers were trained using a sliding window of trajectories.}
        \label{subfig:segments_nonlinear_quartiles}
    \end{subfigure}
     \hspace*{\fill}
    \caption{Quartiles \ac{mse} performance of the learned controllers controlling a nonlinear when compared to the expert controller.}
    \label{fig:nonlinear_quartiles}
\end{figure}
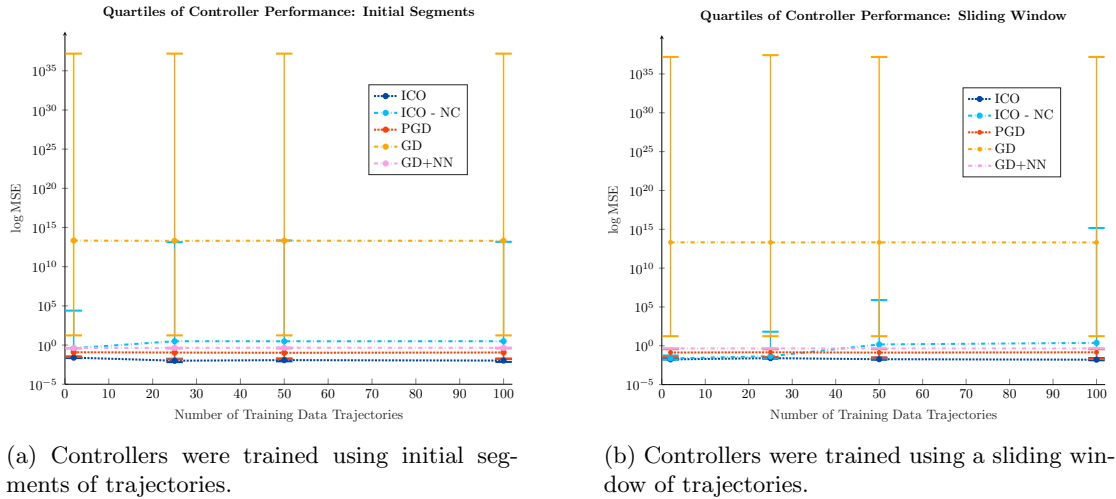

\pagebreak

 
\printbibliography 
\end{document}

%% file: Packages.tex
\usepackage{multicol} 					
\usepackage{color} 							
\usepackage{pgfplots} 					
\usepackage{tikz} 							
\usepackage{amsmath,amssymb,bm,amsthm} 
\usepackage{graphicx,epstopdf}	
\usepackage{ulem}
\usepackage{comment}

\usetikzlibrary{positioning}
\usetikzlibrary{arrows}
\usepackage{wrapfig}

\usepackage{array}
\usepackage{subcaption}
\usepackage{hyperref}
\usepackage[nameinlink]{cleveref} 

\usepackage[backend=bibtex,giveninits=true,dateabbrev=true,style=ieee]{biblatex} 
\appto{\bibsetup}{\raggedright}

\usepackage{algorithm}
\usepackage{algpseudocode}
\usepackage{multirow}
\usepgfplotslibrary{groupplots}
\usepackage{acronym}
\usepackage{tikz}
\usepackage{pgfplots} 
\usepackage[left=0.9in,
            right=0.9in,
            top=1in,
            bottom=1in]{geometry}

%% file: CustomCommands.tex
\def\trans{\mathrm{T}}

\theoremstyle{definition}
\newtheorem{definition}{Definition}[section]

\newtheorem{theorem}{Theorem}
\newtheorem{remark}{Remark}
\newtheorem{lemma}{Lemma}

\newtheorem{corollary}{Corollary}

\newcommand{\Chat}{\widehat{\mathbf{C}}} 
\newcommand{\Ahat}{\widehat{\mathbf{A}}} 
\newcommand{\Bhat}{\widehat{\mathbf{B}}} 
\newcommand{\Dhat}{\widehat{\mathbf{D}}} 
\newcommand{\Atilde}{\tilde{\mathbf{A}}}
\newcommand{\Btilde}{\tilde{\mathbf{B}}}
\newcommand{\Ctilde}{\tilde{\mathbf{C}}}
\newcommand{\Dtilde}{\tilde{\mathbf{D}}}
\newcommand{\I}{\mathbf{I}} 
\newcommand{\mbf}[1]{\mathbf{#1}} 
\newcommand{\bmat}[1]{\begin{bmatrix} #1 \end{bmatrix}} 
\newcommand{\Hcal}{\mathcal{H}} 
\newcommand{\xhat}{\hat{\mathbf{x}}} 
\newcommand{\A}{\mathbf{A}} 
\newcommand{\B}{\mathbf{B}} 
\newcommand{\C}{\mathbf{C}} 
\newcommand{\D}{\mathbf{D}} 
\newcommand{\x}{\mathbf{x}} 
\newcommand{\bu}{\mathbf{u}} 
\newcommand{\y}{\mathbf{y}} 
\newcommand{\yhat}{\hat{\mathbf{y}}} 
\newcommand{\z}{\mathbf{z}} 
\newcommand{\bP}{\mathbf{P}} 
\newcommand{\Q}{\mathbf{Q}} 
\newcommand{\tr}{\mbox{tr}} 
\newcommand{\K}{\mbf{K}} 
\newcommand{\E}{\mbf{E}} 
\newcommand{\bS}{\mbf{S}} 
\newcommand{\R}{\mbf{R}}

\newcommand{\He}{\mbox{He}}
\newcommand{\M}{\mbf{M}}

\newcommand{\W}{\mbf{W}}
\newcommand{\V}{\mbf{V}}

\newcommand{\X}{\mbf{X}}
\newcommand{\Gcal}{\mathcal{G}}

\newcommand{\zero}{\mbf{0}}
\newcommand{\Z}{\mbf{Z}}

\newcommand{\T}{\mbf{T}}

\newcommand{\bU}{\mbf{U}}

\newcommand{\Y}{\mathbf{Y}}
\newcommand{\Ell}{\mathcal{L}}
\newcommand{\nhat}{\hat{n}}
\newcommand{\uhat}{\hat{\bu}}
\newcommand{\optSpace}{\hspace{3mm}}

\newcommand{\br}{\mbf{r}}

\newcommand{\e}{\mbf{e}}
\newcommand{\rhat}{\hat{\br}}

\newcommand{\U}{\mathbf{U}}
\newcommand{\F}{\mathbf{F}}

\newcommand{\loss}{\text{J}}

\DeclareMathOperator*{\argmin}{arg\,min}

\acrodef{dl}[{DL}]{deep learning}
\acrodef{rl}[{RL}]{reinforcement learning}
\acrodef{nn}[{NN}]{neural network}
\acrodef{dnn}[{DNN}]{deep neural network}
\acrodef{tdl}[{TDL}]{temporal difference learning}
\acrodef{pid}[{PID}]{proportional–integral–derivative}
\acrodef{us}[{US}]{Ultrasound}
\acrodef{mse}[{MSE}]{mean squared error}
\acrodef{sgd}[{SGD}]{stochastic gradient descent}
\acrodef{ico}[{ICO}]{iterative convex overbounding}
\acrodef{lmi}[{LMI}]{linear matrix inequality}
\acrodef{mjls}[{MJLS}]{Markov jump linear system}
\acrodef{io}[{IO}]{input-output}
\acrodef{iqc}[{IQC}]{integral quadratic constraints}
\acrodef{cnn}[{CNN}]{convolutional neural network}
\acrodef{il}[{IL}]{Imitation learning}
\acrodef{mpc}[{MPC}]{model predictive control}
\acrodef{nmpc}[{NMPC}]{nonlinear model predictive controller}
\acrodef{sdp}[{SDP}]{semi-definite programming}
\acrodef{relu}[{ReLU}]{rectified linear unit}
\acrodef{us}[US]{ultrasound}
\acrodef{mdp}[MDP]{Markov Decision Process}
\acrodef{iid}[iid]{identical and independently distributed random variable}
\acrodef{pid}[PID]{Proportional Integral Derivative}

\acrodef{lqr}[LQR]{linear-quadratic regulator}

\acrodef{ico}[ICO]{iterative convex overbounding}
\acrodef{pgd}[PGD]{projected gradient descent}
\acrodef{n4sid}[N4SID]{numerical algorithms for subspace state space system identification}
\acrodef{io}[IO]{input-output}
\acrodef{mae}[MAE]{mean absolute error}
\acrodef{mse}[MSE]{mean square error}
\acrodef{gd}[GD]{gradient descent}

\crefname{prob}{Problem}{Problems}
\creflabelformat{prob}{#2{\upshape(#1)}#3} 
\crefname{cons}{Constraint}{Constraints}
\creflabelformat{cons}{#2{\upshape(#1)}#3} 
\crefname{alg}{Algorithm}{Algorithms}
\creflabelformat{alg}{#2{\upshape(#1)}#3}

%% file: Figures/linear_multTraj_median.tex
%
%
\definecolor{mycolor1}{rgb}{0.0,0.28,0.67}%
\definecolor{mycolor2}{rgb}{0.0,0.75,1.0}%
\definecolor{mycolor3}{rgb}{1.0,0.25,0}%
\definecolor{mycolor4}{rgb}{1.0,0.65,0}%
\definecolor{mycolor5}{rgb}{0.98,0.63,0.89}%
\begin{tikzpicture}

\begin{groupplot}[
group style = {group name=my fancy plots,
        group size=1 by 2,
        xticklabels at=edge bottom,
        vertical sep=0pt
    }, width=4in,
    xmin=0, xmax=202
]

\nextgroupplot[ymin = 10^(1.7), ymax = 1e+12,
ytick = {10^(3), 10^6, 10^9, 10^12},
yminorticks=false,
ymode=log,
axis x line*=top,
xmajorticks =false,
x axis line style={draw=none},
axis y discontinuity = parallel,
axis y line=left,
height = 2 in,
title style = {font=\bfseries},
title = {Median Controller Performance: Initial Segments}]

\addplot [color=mycolor5, dashdotted, line width=1.5pt, mark=square, mark options={solid, mycolor5}]
  table[row sep=crcr]{%
2	727.870233808639\\
25	11453.2062563813\\
50	70493.7568508761\\
100	25133729.7308617\\
200	290780.066714642\\
};

\nextgroupplot[ymin = 10^(0), ymax = 10^(2.2),
ymode=log,
ytick = {10^(0), 10^1, 10^2},
yminorticks=false,
height = 2 in,
xmin=0,
xmax=202,
xtick = {0, 25, 50, 75, 100, 200},
xlabel style={font=\color{white!15!black}},
xlabel={Number of Training Data Trajectories},
ylabel style={font=\color{white!15!black}, at={(axis description cs:0,1)}},
ylabel={$\log{\text{MSE}}$},
axis background/.style={fill=white},
axis x line*=bottom,
axis y line*=left,
legend style={at={(0.723,1.6)}, anchor=south west, legend cell align=left, align=left, draw=white!15!black}
]

\addplot [color=mycolor1, densely dotted, line width=1.5pt, mark=o, mark options={solid, mycolor1}]
  table[row sep=crcr]{%
2	14.2826809009194\\
25	16.0974537848286\\
50	16.6806040486284\\
100	17.0289046009463\\
200	16.8581199184947\\
};
\addlegendentry{ICO}

\addplot [color=mycolor3, densely dotted, line width=1.5pt, mark=triangle, mark options={solid, mycolor3}]
  table[row sep=crcr]{%
2	16.832466001678\\
25	15.0915252522591\\
50	16.6026169288553\\
100	17.724428898466\\
200	17.5290888183536\\
};
\addlegendentry{PGD}

\addplot [color=mycolor5, dashdotted, line width=1.5pt, mark=square, mark options={solid, mycolor5}]
  table[row sep=crcr]{%
2	727.870233808639\\
25	11453.2062563813\\
50	70493.7568508761\\
100	25133729.7308617\\
200	290780.066714642\\
};
\addlegendentry{GD+NN}

\end{groupplot}
\end{tikzpicture}%

%% file: Figures/linear_segments_median.tex
%
%

\definecolor{mycolor1}{rgb}{0.0,0.28,0.67}%
\definecolor{mycolor2}{rgb}{0.0,0.75,1.0}%
\definecolor{mycolor3}{rgb}{1.0,0.25,0}%
\definecolor{mycolor4}{rgb}{1.0,0.65,0}%
\definecolor{mycolor5}{rgb}{0.98,0.63,0.89}%
\begin{tikzpicture}

\begin{groupplot}[
group style = {group name=my fancy plots,
        group size=1 by 2,
        xticklabels at=edge bottom,
        vertical sep=0pt
    }, width=4in,
    xmin=0, xmax=202
]

\nextgroupplot[ymin = 10^(1.7), ymax = 1e+12,
ytick = {10^(3), 10^6, 10^9, 10^12},
yminorticks=false,
ymode=log,
axis x line*=top,
xmajorticks =false,
x axis line style={draw=none},
axis y discontinuity = parallel,
axis y line=left,
height = 2 in,
title style = {font=\bfseries},
title = {Median Controller Performance: Sliding Window}]

\addplot [color=mycolor5, dashdotted, line width=1.5pt, mark=square, mark options={solid, mycolor5}]
  table[row sep=crcr]{%
25	463.366003633879\\
50	425254964272.48\\
100	43074766503.5901\\
200	186.216563786309\\
};

\nextgroupplot[ymin = 10^(0), ymax = 10^(2.2),
ymode=log,
ytick = {10^(0), 10^1, 10^2},
yminorticks=false,
height = 2 in,
xmin=0,
xmax=202,
xtick = {0, 25, 50, 75, 100, 200},
xlabel style={font=\color{white!15!black}},
xlabel={Number of Training Data Trajectories},
ylabel style={font=\color{white!15!black}, at={(axis description cs:0,1)}},
ylabel={$\log{\text{MSE}}$},
axis background/.style={fill=white},
axis x line*=bottom,
axis y line*=left,
legend style={at={(0.723,1.6)}, anchor=south west, legend cell align=left, align=left, draw=white!15!black}
]

\addplot [color=mycolor1, densely dotted, line width=1.5pt, mark=o, mark options={solid, mycolor1}]
  table[row sep=crcr]{%
2	17.8122750235148\\
25	16.578936845094\\
50	16.9268634541785\\
100	14.4353289394761\\
200	12.7234383929692\\
};
\addlegendentry{ICO}

\addplot [color=mycolor3, densely dotted, line width=1.5pt, mark=triangle, mark options={solid, mycolor3}]
  table[row sep=crcr]{%
2	13.8103188619606\\
25	14.8445024831058\\
50	14.9805982708387\\
100	15.9695422986655\\
200	15.3658307722247\\
};
\addlegendentry{PGD}

\addplot [color=mycolor5, dashdotted, line width=1.5pt, mark=square, mark options={solid, mycolor5}]
  table[row sep=crcr]{%
2	31.2495784731958\\
25	463.366003633879\\
50	425254964272.48\\
100	43074766503.5901\\
200	186.216563786309\\
};
\addlegendentry{GD+NN}

\end{groupplot}
\end{tikzpicture}%

%% file: Figures/linear_multTraj_max.tex
%
\definecolor{mycolor1}{rgb}{0.0,0.28,0.67}%
\definecolor{mycolor2}{rgb}{0.0,0.75,1.0}%
\definecolor{mycolor3}{rgb}{1.0,0.25,0}%
\definecolor{mycolor4}{rgb}{1.0,0.65,0}%
\definecolor{mycolor5}{rgb}{0.98,0.63,0.89}%

\begin{tikzpicture}

\begin{groupplot}[
group style = {group name=my fancy plots,
        group size=1 by 2,
        xticklabels at=edge bottom,
        vertical sep=0pt
    }, width=4in,
    xmin=0, xmax=202
]

\nextgroupplot[ymin = 10^(5), ymax = 1e+40,
ytick = {10^10, 10^20, 10^30},
yminorticks=false,
ymode=log,
axis x line*=top,
xmajorticks =false,
x axis line style={draw=none},
axis y discontinuity = parallel,
axis y line=left,
height = 2 in,
title style = {font=\bfseries},
title = {Poorest Controller Performance: Initial Segments}]

\addplot [color=mycolor5, dashdotted, line width=1.5pt, mark=square, mark options={solid, mycolor5}]
  table[row sep=crcr]{%
2	1.00462825249782e+32\\
25	3.45057060841858e+19\\
50	1.15964025510405e+19\\
100	1.44101796882365e+15\\
200	363161490077.347\\
};

\nextgroupplot[ymin = 10^(0), ymax = 10^(2.5),
ymode=log,
ytick = {10^(0), 10^1, 10^2},
yminorticks=false,
height = 2 in,
xmin=0,
xmax=205,
xtick = {0, 25, 50, 75, 100, 200},
xlabel style={font=\color{white!15!black}},
xlabel={Number of Training Data Trajectories},
ylabel style={font=\color{white!15!black}, at={(axis description cs:0,1)}},
ylabel={$\log{\text{MSE}}$},
axis background/.style={fill=white},
axis x line=bottom,
axis y line*=left,
legend style={at={(0.68,.05)}, anchor=south west, legend cell align=left, align=left, draw=white!15!black}
]

\addplot [color=mycolor1, densely dotted, line width=1.5pt, mark=o, mark options={solid, mycolor1}]
  table[row sep=crcr]{%
2	27.2179600927491\\
25	26.1084495615442\\
50	25.5894707270926\\
100	26.8708707890281\\
200	27.5288872915486\\
};
\addlegendentry{ICO}

\addplot [color=mycolor3, densely dotted, line width=1.5pt, mark=triangle, mark options={solid, mycolor3}]
  table[row sep=crcr]{%
2	31.0439171967156\\
25	31.3299781644602\\
50	32.5383578486039\\
100	34.4982439398481\\
200	33.7126172692299\\
};
\addlegendentry{PGD}

\addlegendimage{mycolor5, dashdotted, mark = square, mark options = {solid, mycolor5}}
\addlegendentry{GD+NN}

\end{groupplot}
\end{tikzpicture}%

%% file: Figures/linear_segments_max.tex
%
%
\definecolor{mycolor1}{rgb}{0.0,0.28,0.67}%
\definecolor{mycolor2}{rgb}{0.0,0.75,1.0}%
\definecolor{mycolor3}{rgb}{1.0,0.25,0}%
\definecolor{mycolor4}{rgb}{1.0,0.65,0}%
\definecolor{mycolor5}{rgb}{0.98,0.63,0.89}%
\begin{tikzpicture}

\begin{groupplot}[
group style = {group name=my fancy plots,
        group size=1 by 2,
        xticklabels at=edge bottom,
        vertical sep=0pt
    }, width=4in,
    xmin=0, xmax=202
]

\nextgroupplot[ymin = 10^(5), ymax = 1e+40,
ytick = {10^10, 10^20, 10^30},
yminorticks=false,
ymode=log,
axis x line*=top,
xmajorticks =false,
x axis line style={draw=none},
axis y discontinuity = parallel,
axis y line=left,
height = 2 in,
title style = {font=\bfseries},
title = {Poorest Controller Performance: Sliding Window}]

\addplot [color=mycolor5, dashdotted, line width=1.5pt, mark=square, mark options={solid, mycolor5}]
  table[row sep=crcr]{%
2	253246378084.729\\
25	2.49542532448491e+30\\
50	3.09287304252117e+20\\
100	6.70616855268168e+34\\
200	2.07802032468903e+23\\
};

\nextgroupplot[ymin = 10^(0), ymax = 10^(2.5),
ymode=log,
ytick = {10^(0), 10^1, 10^2},
yminorticks=false,
height = 2 in,
xmin=0,
xmax=202,
xtick = {0, 25, 50, 75, 100, 200},
xlabel style={font=\color{white!15!black}},
xlabel={Number of Training Data Trajectories},
ylabel style={font=\color{white!15!black}, at={(axis description cs:0,1)}},
ylabel={$\log{\text{MSE}}$},
axis background/.style={fill=white},
axis x line*=bottom,
axis y line*=left,
legend style={at={(0.68,.05)}, anchor=south west, legend cell align=left, align=left, draw=white!15!black}
]

\addplot [color=mycolor1, densely dotted, line width=1.5pt, mark=o, mark options={solid, mycolor1}]
  table[row sep=crcr]{%
2	25.8926705580996\\
25	27.8674095360142\\
50	30.1149936609428\\
100	29.669708991824\\
200	24.9518635681235\\
};
\addlegendentry{ICO}

\addplot [color=mycolor3, densely dotted, line width=1.5pt, mark=triangle, mark options={solid, mycolor3}]
  table[row sep=crcr]{%
2	23.3916223435961\\
25	23.1675165950881\\
50	22.8841349352065\\
100	22.7872032532922\\
200	21.6894088330167\\
};
\addlegendentry{PGD}

\addlegendimage{mycolor5, dashdotted, mark = square, mark options = {solid, mycolor5}}
\addlegendentry{GD+NN}

\end{groupplot}
\end{tikzpicture}%

%% file: Figures/linear_multTraj_quartiles.tex
%

\definecolor{mycolor1}{rgb}{0.0,0.28,0.67}%
\definecolor{mycolor2}{rgb}{0.0,0.75,1.0}%
\definecolor{mycolor3}{rgb}{1.0,0.25,0}%
\definecolor{mycolor4}{rgb}{1.0,0.65,0}%
\definecolor{mycolor5}{rgb}{0.98,0.63,0.89}%

\begin{tikzpicture}

\begin{axis}[%
width=4.521in,
height=3.566in,
at={(0.758in,0.481in)},
scale only axis,
xmin=0,
xmax=205,
xlabel style={font=\color{white!15!black}},
xlabel={Number of Training Data Trajectories},
ymode=log,
ymin=1e0,
ymax=1e+23,
ytick = {10^(-5), 10^0, 10e+4, 10e+9, 10e+14, 10e+19, 10e+24, 10e+29, 10e+34},
ylabel style={font=\color{white!15!black}},
ylabel={$\log{\text{MSE}}$},
axis background/.style={fill=white},
title style={font=\bfseries},
title={Quartiles of Controller Performance: Initial Segments},
axis x line*=bottom,
axis y line=left,
legend style={at={(0.9,0.85)}, legend cell align=left, align=left, draw=white!15!black}
]
\addplot [color=mycolor1, densely dotted, line width=1.5pt, mark size=1.1pt, mark=*, mark options={solid, fill=mycolor1, draw=mycolor1}]
 plot [error bars/.cd, y dir=both, y explicit, error bar style={solid, line width=1pt}, error mark options={line width=1.5pt, mark size=6.0pt, rotate=90}]
 table[row sep=crcr, y error plus index=2, y error minus index=3]{%
2	14.2826809009194	8.24882831268859	6.35707112655812\\
25	16.0974537848286	5.90673755494369	7.60575794529562\\
50	16.6806040486284	5.51058228834527	9.39698788273399\\
100	17.0289046009463	5.84378279326491	9.83439884170587\\
200	16.8581199184947	5.97249636675732	9.62164002229526\\
};
\addlegendentry{ICO}

\addplot [color=mycolor3, densely dotted, line width=1.5pt, mark size=1.1pt, mark=triangle*, mark options={solid, fill=mycolor3, draw=mycolor3}]
 plot [error bars/.cd, y dir=both, y explicit, error bar style={solid, line width=1pt}, error mark options={line width=1.5pt, mark size=6.0pt, rotate=90}]
 table[row sep=crcr, y error plus index=2, y error minus index=3]{%
2	16.832466001678	5.88606059060557	8.29925401695874\\
25	15.0915252522591	6.58196471459883	4.37855217250159\\
50	16.6026169288553	11.3022023292933	8.0960131564694\\
100	17.724428898466	8.91385001349142	7.99077746923921\\
200	17.5290888183536	8.58445330465925	7.71783402846921\\
};
\addlegendentry{PGD}

\addplot [color=mycolor5, dashdotted, line width=1.5pt, mark size=1.1pt, mark=square*, mark options={solid, fill=mycolor5, draw=mycolor5}]
 plot [error bars/.cd, y dir=both, y explicit, error bar style={solid, line width=1pt}, error mark options={line width=1.5pt, mark size=6.0pt, rotate=90}]
 table[row sep=crcr, y error plus index=2, y error minus index=3]{%
2	727.870233808639	3.51297896093951e+21	577.090817558185\\
25	11453.2062563813	7.3873485174557e+15	11384.9385724907\\
50	70493.7568508761	2161798944.4558	70137.3403978309\\
100	25133729.7308617	36311102663.3148	25131491.265473\\
200	290780.066714642	11790194405.0249	289246.732933428\\
};
\addlegendentry{GD+NN}

\end{axis}
\end{tikzpicture}%

%% file: Figures/linear_segments_quartiles.tex
%
%
\definecolor{mycolor1}{rgb}{0.0,0.28,0.67}%
\definecolor{mycolor2}{rgb}{0.0,0.75,1.0}%
\definecolor{mycolor3}{rgb}{1.0,0.25,0}%
\definecolor{mycolor4}{rgb}{1.0,0.65,0}%
\definecolor{mycolor5}{rgb}{0.98,0.63,0.89}%

\begin{tikzpicture}

\begin{axis}[%
width=4.521in,
height=3.566in,
at={(0.758in,0.481in)},
scale only axis,
xmin=0,
xmax=205,
xlabel style={font=\color{white!15!black}},
xlabel={Number of Training Data Trajectories},
ymode=log,
ymin=1e0,
ymax=1e+23,
ytick = {10^(-5), 10^0, 10e+4, 10e+9, 10e+14, 10e+19, 10e+24, 10e+29, 10e+34},
ylabel style={font=\color{white!15!black}},
ylabel={$\log{\text{MSE}}$},
axis background/.style={fill=white},
title style={font=\bfseries},
title={Quartiles of Controller Performance: Sliding Window},
axis x line*=bottom,
axis y line=left,
legend style={at={(0.9,0.85)}, legend cell align=left, align=left, draw=white!15!black}
]

\addplot [color=mycolor1, densely dotted, line width=1.5pt, mark size=1.1pt, mark=square*, mark options={solid, fill=mycolor1, draw=mycolor1}]
 plot [error bars/.cd, y dir=both, y explicit, error bar style={solid, line width=1pt}, error mark options={line width=1.5pt, mark size=6.0pt, rotate=90}]
 table[row sep=crcr, y error plus index=2, y error minus index=3]{%
2	17.8122750235148	6.09975514575258	8.54366353197334\\
25	16.578936845094	6.49165287004141	8.23381772802889\\
50	16.9268634541785	6.14455014711496	8.39419099929368\\
100	14.4353289394761	9.72696253171427	7.28152224980257\\
200	12.7234383929692	9.27667472206294	5.73711778481236\\
};
\addlegendentry{ICO}

\addplot [color=mycolor3, densely dotted, line width=1.5pt, mark size=1.1pt, mark=square*, mark options={solid, fill=mycolor3, draw=mycolor3}]
 plot [error bars/.cd, y dir=both, y explicit, error bar style={solid, line width=1pt}, error mark options={line width=1.5pt, mark size=6.0pt, rotate=90}]
 table[row sep=crcr, y error plus index=2, y error minus index=3]{%
2	13.8103188619606	7.36715942438213	4.15428751060843\\
25	14.8445024831058	4.84803666166435	4.67205745495011\\
50	14.9805982708387	5.40966814501102	5.22456093505358\\
100	15.9695422986655	4.0798701607842	5.57026095308491\\
200	15.3658307722247	4.44857868825246	5.04349456470738\\
};
\addlegendentry{PGD}

\addplot [color=mycolor5, dashdotted, line width=1.5pt, mark size=1.1pt, mark=square*, mark options={solid, fill=mycolor5, draw=mycolor5}]
 plot [error bars, y dir=both, y explicit, error bar style={solid, line width=1pt}, error mark options={line width=1.5pt, mark size=6.0pt, rotate=90}]
 table[row sep=crcr, y error plus index=2, y error minus index=3]{
2	31.2495784731958	3747.7888380907	11.2456920372142\\
25	463.366003633879	262991978854501	393.347917055966\\
50	425254964272.48	3.46977130948594e+16 425254964209.79\\
100	43074766503.5901	3384691203222.38   43074766420.6059\\
200	186.216563786309	475.518042635575   159.360190775103\\
};
\addlegendentry{GD+NN}
\addplot+[color = mycolor5, mark=-, line width=1.5pt, mark size=6.0pt] coordinates {(50,62.6900024414)};
\draw[color = mycolor5, line width=1pt] (axis cs:50,425254964272.48) -- node[left]{} (axis cs:50,62.6900024414);

\addplot+[color = mycolor5, mark=-, line width=1.5pt, mark size=6.0pt] coordinates {(100,82.9841995239)};
\draw[color = mycolor5, line width=1pt] (axis cs:100,	43074766503.5901) -- node[left]{} (axis cs:100,82.9841995239);

\end{axis}
\end{tikzpicture}%

%% file: Figures/nonlinear_multTraj_median.tex
%
%

\definecolor{mycolor1}{rgb}{0.0,0.28,0.67}%
\definecolor{mycolor2}{rgb}{0.0,0.75,1.0}%
\definecolor{mycolor3}{rgb}{1.0,0.25,0}%
\definecolor{mycolor4}{rgb}{1.0,0.65,0}%
\definecolor{mycolor5}{rgb}{0.98,0.63,0.89}%
\begin{tikzpicture}

\begin{axis}[%
width=4.521in,
height=3.566in,
at={(0.758in,0.481in)},
scale only axis,
xmin=0,
xmax=102,
xlabel style={font=\color{white!15!black}},
xlabel={Number of Training Data Trajectories},
ymode=log,
ymin=1e-05,
ymax=1e+15,
yminorticks=true,
ylabel style={font=\color{white!15!black}},
ylabel={$\log{\text{MSE}}$},
axis background/.style={fill=white},
title style={font=\bfseries},
title={Median Controller Performance: Initial Segments},
axis x line*=bottom,
axis y line=left,
legend style={at={(0.701,0.509)}, anchor=south west, legend cell align=left, align=left, draw=white!15!black}
]
\addplot [color=mycolor1, densely dotted, line width=1.5pt, mark=o, mark options={solid, mycolor1}]
  table[row sep=crcr]{%
2	0.02482942873887\\
25	0.0101296107251832\\
50	0.0120047839111713\\
100	0.0104936458580765\\
};
\addlegendentry{ICO}

\addplot [color=mycolor2, dashdotted, line width=1.5pt, mark=o, mark options={solid, mycolor2}]
  table[row sep=crcr]{%
2	0.397357413854803\\
25	3.10347704347309\\
50	3.0378739478043\\
100	3.09048792345886\\
};
\addlegendentry{ICO - NC}

\addplot [color=mycolor3, densely dotted, line width=1.5pt, mark=triangle, mark options={solid, mycolor3}]
  table[row sep=crcr]{%
2	0.12014326755259\\
25	0.110915623658115\\
50	0.104035605159509\\
100	0.11271966616553\\
};
\addlegendentry{PGD}

\addplot [color=mycolor4, dashdotted, line width=1.5pt, mark=triangle, mark options={solid, mycolor4}]
  table[row sep=crcr]{%
2	21458657316022.8\\
25	19978432548299\\
50	20665714428644.7\\
100	20164742403639.4\\
};
\addlegendentry{GD}

\addplot [color=mycolor5, dashdotted, line width=1.5pt, mark=square, mark options={solid, mycolor5}]
  table[row sep=crcr]{%
2	0.423894768802815\\
25	0.443673898622639\\
50	0.458692470163108\\
100	0.450047110616831\\
};
\addlegendentry{GD+NN}

\end{axis}

\begin{axis}[%
width=5.833in,
height=4.375in,
at={(0in,0in)},
scale only axis,
xmin=0,
xmax=1,
ymin=0,
ymax=1,
axis line style={draw=none},
ticks=none,
axis x line*=bottom,
axis y line*=left
]
\end{axis}
\end{tikzpicture}%

%% file: Figures/nonlinear_segments_median.tex
%
%
\definecolor{mycolor1}{rgb}{0.0,0.28,0.67}%
\definecolor{mycolor2}{rgb}{0.0,0.75,1.0}%
\definecolor{mycolor3}{rgb}{1.0,0.25,0}%
\definecolor{mycolor4}{rgb}{1.0,0.65,0}%
\definecolor{mycolor5}{rgb}{0.98,0.63,0.89}%
\begin{tikzpicture}

\begin{axis}[%
width=4.521in,
height=3.566in,
at={(0.758in,0.481in)},
scale only axis,
xmin=0,
xmax=102,
xlabel style={font=\color{white!15!black}},
xlabel={Number of Training Data Trajectories},
ymode=log,
ymin=1e-05,
ymax=1e+15,
yminorticks=true,
ylabel style={font=\color{white!15!black}},
ylabel={$\log{\text{MSE}}$},
axis background/.style={fill=white},
title style={font=\bfseries},
title={Median Controller Performance: Sliding Window},
axis x line*=bottom,
axis y line=left,
legend style={at={(0.711,0.497)}, anchor=south west, legend cell align=left, align=left, draw=white!15!black}
]
\addplot [color=mycolor1, densely dotted, line width=1.5pt, mark=o, mark options={solid, mycolor1}]
  table[row sep=crcr]{%
2	0.0176113723109893\\
25	0.0256291138841847\\
50	0.0186393929674663\\
100	0.0162123202688351\\
};
\addlegendentry{ICO}

\addplot [color=mycolor2, dashdotted, line width=1.5pt, mark=o, mark options={solid, mycolor2}]
  table[row sep=crcr]{%
2	0.0232249933993789\\
25	0.0454310650803463\\
50	1.49006257545996\\
100	2.38673089019393\\
};
\addlegendentry{ICO - NC}

\addplot [color=mycolor3, densely dotted, line width=1.5pt, mark=triangle, mark options={solid, mycolor3}]
  table[row sep=crcr]{%
2	0.135282505595676\\
25	0.144649975210067\\
50	0.132911116178995\\
100	0.144619849655426\\
};
\addlegendentry{PGD}

\addplot [color=mycolor4, dashdotted, line width=1.5pt, mark=triangle, mark options={solid, mycolor4}]
  table[row sep=crcr]{%
2	21061475429453.3\\
25	20473651515310.4\\
50	20793877413110.5\\
100	20477934600006.3\\
};
\addlegendentry{GD}

\addplot [color=mycolor5, dashdotted, line width=1.5pt, mark=square, mark options={solid, mycolor5}]
  table[row sep=crcr]{%
2	0.462749700817854\\
25	0.456047573957067\\
50	0.45980120445646\\
100	0.456414973649588\\
};
\addlegendentry{GD+NN}

\end{axis}

\begin{axis}[%
width=5.833in,
height=4.375in,
at={(0in,0in)},
scale only axis,
xmin=0,
xmax=1,
ymin=0,
ymax=1,
axis line style={draw=none},
ticks=none,
axis x line*=bottom,
axis y line*=left
]
\end{axis}
\end{tikzpicture}%

%% file: Figures/nonlinear_multTraj_max.tex
%
%
\definecolor{mycolor1}{rgb}{0.0,0.28,0.67}%
\definecolor{mycolor2}{rgb}{0.0,0.75,1.0}%
\definecolor{mycolor3}{rgb}{1.0,0.25,0}%
\definecolor{mycolor4}{rgb}{1.0,0.65,0}%
\definecolor{mycolor5}{rgb}{0.98,0.63,0.89}%
\begin{tikzpicture}

\begin{groupplot}[
group style = {group name=my fancy plots,
        group size=1 by 2,
        xticklabels at=edge bottom,
        vertical sep=0pt
    }, width=4in,
    xmin=0, xmax=102
]

\nextgroupplot[ymin = 1e+80, ymax = 1e+205,
ytick = {1e+100, 1e+150, 1e+200},
yminorticks=true,
ymode=log,
axis x line*=top,
xmajorticks =false,
x axis line style={draw=none},
axis y discontinuity = parallel,
axis y line=left,
height = 2 in,
title style = {font=\bfseries},
title = {Poorest Controller Performance: Initial Segments}]
\addplot [color=mycolor2, dashdotted, line width=1.5pt, mark=o, mark options={solid, mycolor2}]
  table[row sep=crcr]{%
2	6.42765178382836e+161\\
25	1.33627846984799e+141\\
50	1.14351703102895e+141\\
100	3.79869075142263e+141\\
};

\addplot [color=mycolor4, dashdotted, line width=1.5pt, mark=triangle, mark options={solid, mycolor4}]
  table[row sep=crcr]{%
2	3.03941299791684e+169\\
25	3.04141316038215e+169\\
50	3.04221359879817e+169\\
100	3.04261389810701e+169\\
};

\nextgroupplot[ymin = 10^(-1.75), ymax = 10^(1.1),
ymode=log,
ytick = {10^(-1.5), 10^(-1), 10^(-.5), 1e+0, 10^(0.5), 10^1},
yminorticks=true,
height = 2 in,
xmin=0,
xmax=102,
xtick = {0, 25, 50, 75, 100},
xlabel style={font=\color{white!15!black}},
xlabel={Number of Training Data Trajectories},
ylabel style={font=\color{white!15!black}, at={(axis description cs:-0.025,1)}},
ylabel={$\log{\text{MSE}}$},
axis background/.style={fill=white},
axis x line*=bottom,
axis y line*=left,
legend style={at={(0.65,0.78)}, anchor=south west, legend cell align=left, align=left, draw=white!15!black}
]

\addplot [color=mycolor1, densely dotted, line width=1.5pt, mark=o, mark options={solid, mycolor1}]
  table[row sep=crcr]{%
2	0.135732951253067\\
25	0.675334449216096\\
50	0.669314403627093\\
100	0.681468326066919\\
};
\addlegendentry{ICO}

\addlegendimage{mycolor2, dashdotted, mark = o, mark options = {solid, mycolor2}}
\addlegendentry{ICO - NC}

\addplot [color=mycolor3, densely dotted, line width=1.5pt, mark=triangle, mark options={solid, mycolor3}]
  table[row sep=crcr]{%
2	1.1388810065229\\
25	1.10991890578301\\
50	1.11349041873625\\
100	1.10991890736564\\
};
\addlegendentry{PGD}

\addlegendimage{mycolor4, dashdotted, mark = triangle, mark options = {solid, mycolor4}}
\addlegendentry{GD}

\addplot [color=mycolor5, dashdotted, line width=1.5pt, mark=square, mark options={solid, mycolor5}]
  table[row sep=crcr]{%
2	0.479356531473988\\
25	0.55901287416767\\
50	0.469902387276518\\
100	0.463758349541128\\
};
\addlegendentry{GD+NN}

\end{groupplot}

\end{tikzpicture}

%% file: Figures/nonlinear_segments_max.tex
%
%

\definecolor{mycolor1}{rgb}{0.0,0.28,0.67}%
\definecolor{mycolor2}{rgb}{0.0,0.75,1.0}%
\definecolor{mycolor3}{rgb}{1.0,0.25,0}%
\definecolor{mycolor4}{rgb}{1.0,0.65,0}%
\definecolor{mycolor5}{rgb}{0.98,0.63,0.89}%
\begin{tikzpicture}

\begin{groupplot}[
group style = {group name=my fancy plots,
        group size=1 by 2,
        xticklabels at=edge bottom,
        vertical sep=0pt
    }, width=4in,
    xmin=0, xmax=102
]

\nextgroupplot[ymin = 1e+80, ymax = 1e+205,
ytick = {1e+100, 1e+150, 1e+200},
yminorticks=true,
ymode=log,
axis x line*=top,
xmajorticks =false,
x axis line style={draw=none},
axis y discontinuity = parallel,
axis y line=left,
height = 2 in,
title style = {font=\bfseries},
title = {Poorest Controller Performance: Sliding Window}]
\addplot [color=mycolor2, dashdotted, line width=1.5pt, mark=o, mark options={solid, mycolor2}]
  table[row sep=crcr]{%
2	1.23072692648245e+108\\
25	4.4696752797084e+154\\
50	5.18504161877824e+181\\
100	1.97971247571973e+147\\
};

\addplot [color=mycolor4, dashdotted, line width=1.5pt, mark=triangle, mark options={solid, mycolor4}]
  table[row sep=crcr]{%
2	3.03941299791684e+169\\
25	3.0315566427385e+169\\
50	3.03781382381257e+169\\
100	3.03781382381257e+169\\
};

\nextgroupplot[ymin = 10^(-1.75), ymax = 10^(1.1),
ymode=log,
ytick = {10^(-1.5), 10^(-1), 10^(-.5), 1e+0, 10^(0.5), 10^1},
yminorticks=true,
height = 2 in,
xmin=0,
xmax=102,
xtick = {0, 25, 50, 75, 100},
xlabel style={font=\color{white!15!black}},
xlabel={Number of Training Data Trajectories},
ylabel style={font=\color{white!15!black}, at={(axis description cs:-0.025,1)}},
ylabel={$\log{\text{MSE}}$},
axis background/.style={fill=white},
axis x line*=bottom,
axis y line*=left,
legend style={at={(0.65,0.875)}, anchor=south west, legend cell align=left, align=left, draw=white!15!black}
]

\addplot [color=mycolor1, densely dotted, line width=1.5pt, mark=o, mark options={solid, mycolor1}]
  table[row sep=crcr]{%
2	0.0288877209552402\\
25	0.063387350614706\\
50	0.244091415390301\\
100	0.311659307914939\\
};
\addlegendentry{ICO}

\addlegendimage{mycolor2, dashdotted, mark = o, mark options = {solid, mycolor2}}
\addlegendentry{ICO - NC}

\addplot [color=mycolor3, densely dotted, line width=1.5pt, mark=triangle, mark options={solid, mycolor3}]
  table[row sep=crcr]{%
2	4.19259044167802\\
25	4.19259044131716\\
50	4.19259130670499\\
100	4.17518974913209\\
};
\addlegendentry{PGD}

\addlegendimage{mycolor4, dashdotted, mark = triangle, mark options = {solid, mycolor4}}
\addlegendentry{GD}

\addplot [color=mycolor5, dashdotted, line width=1.5pt, mark=square, mark options={solid, mycolor5}]
  table[row sep=crcr]{%
2	0.576425409190423\\
25	0.496403598785122\\
50	0.514447889442841\\
100	0.489317724241665\\
};
\addlegendentry{GD+NN}

\end{groupplot}

\end{tikzpicture}

%% file: Figures/nonlinear_multTraj_quartiles.tex
%
\definecolor{mycolor1}{rgb}{0.0,0.28,0.67}%
\definecolor{mycolor2}{rgb}{0.0,0.75,1.0}%
\definecolor{mycolor3}{rgb}{1.0,0.25,0}%
\definecolor{mycolor4}{rgb}{1.0,0.65,0}%
\definecolor{mycolor5}{rgb}{0.98,0.63,0.89}%
\begin{tikzpicture}

\begin{axis}[%
width=4.521in,
height=3.566in,
at={(0.758in,0.481in)},
scale only axis,
xmin=0,
xmax=102,
xlabel style={font=\color{white!15!black}},
xlabel={Number of Training Data Trajectories},
ymode=log,
ymin=1e-5,
ymax=1e+40,
ytick = {10^(-5), 10^0, 10e+4, 10e+9, 10e+14, 10e+19, 10e+24, 10e+29, 10e+34},
ylabel style={font=\color{white!15!black}},
ylabel={$\log{\text{MSE}}$},
axis background/.style={fill=white},
title style={font=\bfseries},
title={Quartiles of Controller Performance: Initial Segments},
axis x line*=bottom,
axis y line=left,
legend style={at={(0.9,0.85)}, legend cell align=left, align=left, draw=white!15!black}
]

\addplot [color=mycolor1, densely dotted, line width=1.5pt, mark size=1.5pt, mark=*, mark options={solid, fill=mycolor1, draw=mycolor1}]
 plot [error bars/.cd, y dir=both, y explicit, error bar style={solid, line width=1pt}, error mark options={line width=1.5pt, mark size=6.0pt, rotate=90}]
 table[row sep=crcr, y error plus index=2, y error minus index=3]{%
2	0.02482942873887	0.00487144952367508	0.00285100473245775\\
25	0.0101296107251832	0.00641106964290917	0.0034911614480575\\
50	0.0120047839111713	0.00525107177145509	0.00358692172478532\\
100	0.0104936458580765	0.0064262417691485	0.00340450943066058\\
};
\addlegendentry{ICO}

\addplot [color=mycolor2, dashdotted, line width=1.5pt, mark size=1.5pt, mark=*, mark options={solid, fill=mycolor2, draw=mycolor2}]
 plot [error bars/.cd, y dir=both, y explicit, error bar style={solid, line width=1pt}, error mark options={line width=1.5pt, mark size=6.0pt, rotate=90}]
 table[row sep=crcr, y error plus index=2, y error minus index=3]{%
2	0.397357413854803	24830.437671401	0.364109917537241\\
25	3.10347704347309	13546132984071.8	3.08469770145104\\
50	3.0378739478043	22267952444410.8	3.01805542599625\\
100	3.09048792345886	15050606071653.6	3.07061837511872\\
};
\addlegendentry{ICO - NC}

\addplot [color=mycolor3, densely dotted, line width=1.5pt, mark size=1.5pt, mark=*, mark options={solid, fill=mycolor3, draw=mycolor3}]
 plot [error bars/.cd, y dir=both, y explicit, error bar style={solid, line width=1pt}, error mark options={line width=1.5pt, mark size=6.0pt, rotate=90}]
 table[row sep=crcr, y error plus index=2, y error minus index=3]{%
2	0.12014326755259	0.27045335882781	0.0852707063955383\\
25	0.110915623658115	0.279591406484391	0.0958588718981471\\
50	0.104035605159509	0.286466996408016	0.0863576919179863\\
100	0.11271966616553	0.277782783844439	0.0962449362039551\\
};
\addlegendentry{PGD}

\addplot [color=mycolor4, dashdotted, line width=1.5pt, mark size=1.5pt, mark=*, mark options={solid, fill=mycolor4, draw=mycolor4}]
 plot [error bars/.cd, y dir=both, y explicit, error bar style={solid, line width=1pt}, error mark options={line width=1.5pt, mark size=6.0pt, rotate=90}]
table[row sep=crcr, y error plus index=2, y error minus index=3]{%
2	21458657316022.8  1.55716249121865e+37	21458657316005.6\\
25	19978432548299	1.55707334079053e+37	19978432548281.8\\
50	20665714428644.7	1.55706659263591e+37	20665714428627.5\\
100	20164742403639.4	1.55707334079053e+37	20164742403622.2\\
};
\addlegendentry{GD}

\addplot [color=mycolor5, dashdotted, line width=1.5pt, mark size=1.5pt, mark=*, mark options={solid, fill=mycolor5, draw=mycolor5}]
 plot [error bars/.cd, y dir=both, y explicit, error bar style={solid, line width=1pt}, error mark options={line width=1.5pt, mark size=6.0pt, rotate=90}]
  table[row sep=crcr, y error plus index=2, y error minus index=3]{%
2	0.423894768802815	0.00421136684321516	0.00657876625284187\\
25	0.443673898622639	0.00231493583077957	0.00237878813562509\\
50	0.458692470163108	0.00249582820409339	0.00292635869630181\\
100	0.450047110616831	0.00267951576855008	0.00240056137098521\\
};
\addlegendentry{GD+NN}

\addplot+[color = mycolor4, mark=-, line width=1.5pt, mark size=6.0pt] coordinates {(2,17.19921875)};
\draw[color = mycolor4, line width=1pt] (axis cs:2,21458657316022.8) -- node[left]{} (axis cs:2,17.19921875);

\addplot+[color = mycolor4, mark=-, line width=1.5pt, mark size=6.0pt] coordinates {(25,17.19921875)};
\draw[color = mycolor4, line width=1pt] (axis cs:25,1997843254829) -- node[left]{} (axis cs:25,17.19921875);

\addplot+[color = mycolor4, mark=-, line width=1.5pt, mark size=6.0pt] coordinates {(50,17.19921875)};
\draw[color = mycolor4, line width=1pt] (axis cs:50,20665714428644.7) -- node[left]{} (axis cs:50,17.19921875);

\addplot+[color = mycolor4, mark=-, line width=1.5pt, mark size=6.0pt] coordinates {(100,17.19921875)};
\draw[color = mycolor4, line width=1pt] (axis cs:100,20164742403639.4) -- node[left]{} (axis cs:100,17.19921875);

\end{axis}
\end{tikzpicture}%

%% file: Figures/nonlinear_segments_quartiles.tex
%
%
\definecolor{mycolor1}{rgb}{0.0,0.28,0.67}%
\definecolor{mycolor2}{rgb}{0.0,0.75,1.0}%
\definecolor{mycolor3}{rgb}{1.0,0.25,0}%
\definecolor{mycolor4}{rgb}{1.0,0.65,0}%
\definecolor{mycolor5}{rgb}{0.98,0.63,0.89}%
\begin{tikzpicture}

\begin{axis}[%
width=4.521in,
height=3.566in,
at={(0.758in,0.481in)},
scale only axis,
xmin=0,
xmax=102,
xlabel style={font=\color{white!15!black}},
xlabel={Number of Training Data Trajectories},
ymode=log,
ymin=1e-5,
ymax=1e+40,
ytick = {10^(-5), 10^0, 10e+4, 10e+9, 10e+14, 10e+19, 10e+24, 10e+29, 10e+34},
ylabel style={font=\color{white!15!black}},
ylabel={$\log{\text{MSE}}$},
axis background/.style={fill=white},
title style={font=\bfseries},
title={Quartiles of Controller Performance: Sliding Window},
axis x line*=bottom,
axis y line=left,
legend style={at={(0.9,0.85)}, legend cell align=left, align=left, draw=white!15!black}
]

\addplot [color=mycolor1, densely dotted, line width=1.5pt, mark size=1.5pt, mark=*, mark options={solid, fill=mycolor1, draw=mycolor1}]
 plot [error bars/.cd, y dir=both, y explicit, error bar style={solid, line width=1pt}, error mark options={line width=1.5pt, mark size=6.0pt, rotate=90}]
 table[row sep=crcr, y error plus index=2, y error minus index=3]{%
2	0.0176113723109893	0.00545116937397688	0.000134177750660078\\
25	0.0256291138841847	0.00273928405031977	0.0032059904663661\\
50	0.0186393929674663	0.0111868171392488	0.00254395803082328\\
100	0.0162123202688351	0.0100689464714574	0.00263103849858718\\
};
\addlegendentry{ICO}

\addplot [color=mycolor2, dashdotted, line width=1.5pt, mark size=1.5pt, mark=*, mark options={solid, fill=mycolor2, draw=mycolor2}]
 plot [error bars/.cd, y dir=both, y explicit, error bar style={solid, line width=1pt}, error mark options={line width=1.5pt, mark size=6.0pt, rotate=90}]
 table[row sep=crcr, y error plus index=2, y error minus index=3]{%
2	0.0232249933993789	0.0335524676003064	0.00561763659129204\\
25	0.0454310650803463	64.785451652711	0.0183308235172397\\
50	1.49006257545996	774472.46376916	1.46645998662023\\
100	2.38673089019393	1.48407992351784e+15	2.36620953787634\\
};
\addlegendentry{ICO - NC}

\addplot [color=mycolor3, densely dotted, line width=1.5pt, mark size=1.0pt, mark=*, mark options={solid, fill=mycolor3, draw=mycolor3}]
 plot [error bars/.cd, y dir=both, y explicit, error bar style={solid, line width=1pt}, error mark options={line width=1.5pt, mark size=6.0pt, rotate=90}]
 table[row sep=crcr, y error plus index=2, y error minus index=3]{%
2	0.135282505595676	0.232329652705264	0.101334295847699\\
25	0.144649975210067	0.219828986067196	0.110045009363471\\
50	0.132911116178995	0.233594550589332	0.10248843862925\\
100	0.144619849655426	0.220454115595446	0.120716079526013\\
};
\addlegendentry{PGD}

\addplot [color=mycolor4, dashdotted, line width=1.5pt, mark size=1.0pt, mark=*, mark options={solid, fill=mycolor4, draw=mycolor4}]
 plot [error bars/.cd, y dir=both, y explicit, error bar style={solid, line width=1pt}, error mark options={line width=1.5pt, mark size=6.0pt, rotate=90}]
table[row sep=crcr, y error plus index=2, y error minus index=3]{%
2	21061475429453.3	1.55732643288339e+37	21061475429436.1\\
25	20473651515310.4	2.64157822848175e+37	20473651515293.1\\
50	20793877413110.5	1.55747530963169e+37	20793877413093.3\\
100	20477934600006.3	1.55756462859205e+37	20477934599989.1\\
};
\addlegendentry{GD}

\addplot [color=mycolor5, dashdotted, line width=1.5pt, mark size=0.4pt, mark=*, mark options={solid, fill=mycolor5, draw=mycolor5}]
 plot [error bars/.cd, y dir=both, y explicit, error bar style={solid, line width=1pt}, error mark options={line width=1.5pt, mark size=6.0pt, rotate=90}]
  table[row sep=crcr, y error plus index=2, y error minus index=3]{%
2	0.462749700817854	0.0264238994201181	0.020631542323114\\
25	0.456047573957067	0.00703851522980048	0.0068824072303742\\
50	0.45980120445646	0.00424990313906171	0.00540685398554464\\
100	0.456414973649588	0.0057878911332504	0.00399484231801528\\
};
\addlegendentry{GD+NN}

\addplot+[color = mycolor4, mark=-, line width=1.5pt, mark size=6.0pt] coordinates {(2,17.19921875)};
\draw[color = mycolor4, line width=1pt] (axis cs:2,21061475429453.3) -- node[left]{} (axis cs:2,17.19921875);

\addplot+[color = mycolor4, mark=-, line width=1.5pt, mark size=6.0pt] coordinates {(25,17.296875)};
\draw[color = mycolor4, line width=1pt] (axis cs:25,20473651515310.4) -- node[left]{} (axis cs:25,17.296875);

\addplot+[color = mycolor4, mark=-, line width=1.5pt, mark size=6.0pt] coordinates {(50,17.19921875)};
\draw[color = mycolor4, line width=1pt] (axis cs:50,20793877413110.5) -- node[left]{} (axis cs:50,17.19921875);

\addplot+[color = mycolor4, mark=-, line width=1.5pt, mark size=6.0pt] coordinates {(100,17.19921875)};
\draw[color = mycolor4, line width=1pt]  (axis cs:100,20477934600006.3) -- node[left]{} (axis cs:100,17.19921875);
\end{axis}
\end{tikzpicture}

%% file: arxivDissILJournal.bib
@string { ACC = "Amer. Control Conf."}

@string { ARMA = "Archive Rational Mechanics and Analysis "}

@string { AUT = "Automatica" }

@string { ECC = "European Control Conf."}

@string { IJC = "Int. J. Control"}

@string { SCL = "Systems \& Control Letters"}

@String{CDC61 = "61st IEEE Conference on Decision and Control"}

@string{TAC 			   	= "IEEE Tran. Aut. Ctrl."}

@string{IFACpapersonline  	= "IFAC-PapersOnline"}

@string{AUT				= "Automatica"}

@string{IJC				= "Int. J. Ctrl."}

@string{CDC21			= "21st IEEE Conf. Decis. Ctrl."}

@string{CDC61			= "61st IEEE Conf. Decis. Ctrl."}

@string{TCS				= "IEEE Tran. Circ. Sys."}

@string{SCL				= "Sys. Ctrl. Lett."}

@string{RNC				= "Int. J. Robust Nonlin. Ctrl."}

@string{ArXiv				= "Ar$\chi$iv"}

@string{MP     = "Math. Prog."}

@string{ARMA = "Arch. Rational Mech. Anal."}

@string{ECC = "Euro. Ctrl. Conf."}

@string{PMLR = "Proc. Mach. Learn"}

@string{L4DC2 = "Proc. 2nd Conf. Learn. Dyn. Control (L4DC)" }

@string{L4DC3 = "Proc. 3rd Conf. Learn. Dyn. Control (L4DC)" }

@string{L4DC4 = "Proc. 4th Conf. Learn. Dyn. Control (L4DC)" }

@string{ICLR = "Int. Conf. Learn. Represent."}

@string{IEEECSL = "IEEE Control Syst. Lett."}

@string{JORS = "J. Oper. Res. Soc."}

@string{MP = "Math. Program."}

@string{JOTA = "J. Optim."}

@string{ICRA = "IEEE Int. Conf. Robot. Autom."}

@article{Zames1966,
	author = {Zames, G.},
	journal = TAC,
	number = {2--3},
	title = "On the Input-Output Stability of Time-Varying Nonlinear Feedback Systems Parts {I \& II}",
	volume = {\textsc{ac}-11},
	year = {1966}
}

@article{Warner2017a,
	author = {Warner, E. C. and Scruggs, J. T.},
	journal = IFACpapersonline,
	volume = {50},
	number = {1},
	pages = {10449--10455},
	title = "Iterative Convex Overbounding Algorithms for {BMI} Optimization Problems",
	year = {2017}
}

@article{Bridgeman2014,
	author = {Bridgeman, Leila Jasmine and Forbes, James Richard},
	journal = IJC,
	number = {8},
	pages = {1467--1477},
	title = "Conic-sector-based control to circumvent passivity violations",
	volume = {87},
	year = {2014}
}

@article{Bridgeman2016,
	author = {Bridgeman, Leila Jasmine and Forbes, James Richard},
	journal = TAC,
	volume = {61},
	number = {7},
	pages = {1931--1937},
	title = "The Extended Conic Sector Theorem",
	month = jul,
	year = {2016}
}

@article{Sivaranjani2018,
	author = {Sivaranjani, S. and Forbes, James Richard and Seiler, Peter and Gupta, Vijay},
	journal = TAC,
	volume = {2},
	number = {2},
	title = "Conic-Sector-Based Analysis and Control Synthesis for LInear Parameter Varying Systems",
	month = apr,
	year = {2018}
}

@article{Xia2020,
	author = {Xia, Meng and Gahinet, Pascal and Abroug, Neil and Buhr, Craig and Laroche, Edouard},
	title = "Sector bounds in stability analysis and control design",
	journal = RNC,
	month = may,
	volume = {30},
	pages = {7857--7882},
	year = {2020}
}

@book{Brogliato2007,
	address = {London, UK},
	author = {Brogliato, Bernard and Lozano, Rogelio and Maschke, Bernhard and Egeland, Olav},
	edition = {2},
	publisher = {Springer Verlag},
	title = "Dissipative Systems Analysis and Control: Theory and Applications",
	year = {2007}
}

@article{Willems1972,
	author = {Willems, Jan C.},
	title = "Dissipative Dynamical Systems Part {I}: General Theory",
	journal = ARMA,
	pages = {321-351},
	volume = 45,
	year = {1972}
}

@article{Hill1977,
	author = {Hill, D. J. and Moylan, P. J.},
	title = "Stability Results for Nonlinear Feedback Systems",
	journal = AUT,
	pages = {377-382},
	volume = 13,
	number = 4,
	month = jul,
	year = {1977}
}

@article{Hill1976,
	author = {Hill, David and Moylan, Peter},
	title = "The Stability of Nonlinear Dissipative Systems",
	journal = TAC,
	pages = {377-382},
	volume = 21,
	number = 5,
	month = oct,
	year = {1976}
}

@book{Desoer1975,
	address = {New York, NY},
	author = {Desoer, C. A. and Vidyasagar, M.},
	edition = {2},
	publisher = {Academic Press, Inc.},
	title = "Feedback Systems: Input-Output Properties",
	year = {1975}
}

@article{Forbes2019,
	author = {Forbes, James Richard},
	title = "Synthesis of strictly positive real {$\Hcal_2$} controllers using dialated {LMI}s",
	journal = IJC,
	pages = {2584-2590},
	volume = 92,
	number = 11,
	year = {2019}
}

@article{Vidyasagar1977,
	author = {Vidyasagar, M.},
	title = "{$\Ell_2$}-Stability of interconnected systems using a reformulation of the passivity theorem",
	journal = TCS,
	pages = {637-645},
	volume = "cas-24",
	number = 11,
	month = nov,
	year = {1977}
}

@article{CaverlyArXiv,
	author = {Caverly, Ryan James and Forbes, James Richard},
	title = "{LMI} Properties and applications in systems, stability, and control theory",
	journal = ArXiv,
	pages = {1-157},
	number = 3,
	month = apr,
	year = {2021},
	url = {arXiv:1903.08599v3}
}

@book{Safonov1980,
	author = {Safonov},
	title = "Stability and Robustness of multivariable feedback systems",
	publisher = {MIT Press},
	address = {Cambridge, MA},
	year = {1980}
}

@book{Vidyasagar1981,
	author = {Vidyasagar, M.},
	title = "Input-output analysis of large-scale interconnected systems",
	publisher = {Springer-Verlag},
	location = {Berlin, Germany},
	year = {1981}
}

@article{Megretski1997,
	author = {Megretski, A. and Rantzer, A.},
	title = "System analysis via integral quadratic constraints",
	journal = TAC,
	pages = {819-830},
	volume = 42,
	number = 6,
	year = {1997}
}

@article{Scorletti2001,
	author = {Scorletti, G{\'{e}}rard and Duc, Gilles},
	title = "An {LMI} approach to decentralized {$H_{\infty}$} control",
	journal = IJC,
	pages = {211-224},
	volume = 74,
	number = 3,
	year = {2001}
}

@article{Havens2021,
	author = {Havens, Aaron and Hu, Bin},
	title = "On estimation learning of linear control policies: enforcing stability and robustness constraints via {LMI} conditions",
	journal = ACC,
	month = May,
	location = "New Orleans, USA",
	pages = {882-887},
	year = {2021}
}

@article{Geromel1997,
	author = {Geromel, J. C. and Gapski, P. B},
	title = "Synthesis of positive real {$\mathcal{H}_2$ controllers}",
	journal = TAC,
	month = jul,
	volume = {42},
	number = {7},
	pages = {988-992},
	year = {1997}
}

@article{Palan2020,
  title = 	 {Fitting a Linear Control Policy to Demonstrations with a Kalman Constraint},
  author =       {Palan, Malayandi and Barratt, Shane and McCauley, Alex and Sadigh, Dorsa and Sindhwani, Vikas and Boyd, Stephen},
  journal = 	L4DC2,
  pages = 	 {374--383},
  year = 	 {2020},
  volume = 	 {120},
  series = 	 PMLR
  
}

@article{Makdah2021,
  author    = {Abed AlRahman Al Makdah and
               Vishaal Krishnan and
               Fabio Pasqualetti},
  title     = {Learning Robust Feedback Policies from Demonstrations},
  journal   = {Ar$\chi$iv},
  year      = {2021},
  eprinttype = {arXiv},
  eprint    = {2103.16629}}

@ARTICLE{Yin2022,
  author={Yin, He and Seiler, Peter and Jin, Ming and Arcak, Murat},
  journal= IEEECSL, 
  title={Imitation Learning With Stability and Safety Guarantees}, 
  year={2022},
  volume={6},
  number={},
  pages={409-414}}

@BOOK{Osa2018,
  author={Osa, Takayuki and Pajarinen, Joni and Neumann, Gerhard and Bagnell, J. Andrew and Abbeel, Pieter and Peters, Jan},
  title={An Algorithmic Perspective on Imitation Learning},
  year={2017},
  publisher={Found. and Trends in Robotics},
  volume={7},
  number={1-2},
  pages={1-179}}

@article{Pauli2021,
  title = 	 {Offset-free setpoint tracking using neural network controllers},
  author =       {Pauli, Patricia and K\"ohler, Johannes and Berberich, Julian and Koch, Anne and Allg\"ower, Frank},
  journal = 	 L4DC3,
  pages = 	 {992--1003},
  year = 	 {2021},
  volume = 	 {144},
  series = 	 PMLR}

@article{Chen2018,
  author={Chen, Steven and Saulnier, Kelsey and Atanasov, Nikolay and Lee, Daniel D. and Kumar, Vijay and Pappas, George J. and Morari, Manfred},
  journal=ACC, 
  title={Approximating Explicit Model Predictive Control Using Constrained Neural Networks}, 
  year={2018},
  volume={},
  number={},
  pages={1520-1527}
  }

@article{Revay2020,
  title = 	 {Contracting Implicit Recurrent Neural Networks: Stable Models with Improved Trainability},
  author =       {Revay, Max and Manchester, Ian},
  journal = 	 L4DC2,
  pages = 	 {393--403},
  year = 	 {2020},
  volume = 	 {120},
  series = 	 PMLR
}

@article{Donti2021,
title={Enforcing robust control guarantees within neural network policies},
author={Priya L. Donti and Melrose Roderick and Mahyar Fazlyab and J Zico Kolter},
journal=ICLR,
year={2021}
}

@ARTICLE{Byrnes1994,
  author={Byrnes, C.I. and Wei Lin},
  journal=TAC, 
  title={Losslessness, feedback equivalence, and the global stabilization of discrete-time nonlinear systems}, 
  year={1994},
  volume={39},
  number={1},
  pages={83-98},
  publisher={IEEE}}

@article{Kottenstette2010,
  author={Kottenstette, Nicholas and Antsaklis, Panos J.},
  journal = ACC, 
  title={Relationships between positive real, passive dissipative, \& positive systems}, 
  year={2010},
  volume={},
  number={},
  pages={409-416}}

@article{Hitz,
title = {Discrete positive-real funtions and their application to system stability},
journal = {Proceedings of the Institution of Electrical Engineering},
volume = {116},
number = {1},
pages = {153-155},
year = {1969},
author = {L. Hitz and B.D.). Anderson}
}

@article{Strong2022,
  title={Dissipative Imitation Learning for Robust Dynamic Output Feedback},
  author={Strong, Amy K. and Lo{C}icero, Ethan J. and Bridgeman, Leila},
  journal=CDC61,
  year={2022}
}

@article{haddad2004vector,
  title={Vector dissipativity theory for discrete-time large-scale nonlinear dynamical systems},
  author={Haddad, Wassim M and Hui, Qing and Chellaboina, Vijaysekhar and Nersesov, Sergey},
  journal={Advances in Difference Equations},
  volume={2004},
  number={1},
  pages={1--30},
  year={2004},
  publisher={SpringerOpen}
}

@article{carrasco2013towards,
  title={Towards L2-stability of discrete-time reset control systems via dissipativity theory},
  author={Carrasco, Joaqu{\'\i}n and Navarro-L{\'o}pez, Eva M},
  journal=SCL,
  volume={62},
  number={6},
  pages={525--530},
  year={2013},
  publisher={Elsevier}
}

@article{van1994n4sid,
  title={N4SID: Subspace algorithms for the identification of combined deterministic-stochastic systems},
  author={Van Overschee, Peter and De Moor, Bart},
  journal={Automatica},
  volume={30},
  number={1},
  pages={75--93},
  year={1994},
  publisher={Elsevier}
}

@article{kang2022lyapunov,
  title={Lyapunov density models: Constraining distribution shift in learning-based control},
  author={Kang, Katie and Gradu, Paula and Choi, Jason J and Janner, Michael and Tomlin, Claire and Levine, Sergey},
  journal={International Conference on Machine Learning},
  pages={10708--10733},
  year={2022},
  organization={PMLR}
}

@article{sebe2018sequential,
  title={Sequential convex overbounding approximation method for bilinear matrix inequality problems},
  author={Sebe, Noboru},
  journal={IFAC-PapersOnLine},
  volume={51},
  number={25},
  pages={102--109},
  year={2018},
  publisher={Elsevier}
}

@article{de2000convexifying,
    author = {{De Oliveira}, M. C. and Camino, J. F. and Skelton, R. E.},
    journal = {39th IEEE Conf. Decis. Ctrl.},
    volume = {3},
    pages = {2781--2786},
    title = "A convexifying algorithm for the design of structured linear controllers",
    month = dec,
    location = {Sydney, Aus},
    year = {2000}
}

@book{boyd1994linear,
  title={Linear matrix inequalities in system and control theory},
  author={Boyd, Stephen and El Ghaoui, Laurent and Feron, Eric and Balakrishnan, Venkataramanan},
  year={1994},
  publisher={SIAM}
}

@article{zhou1988robust,
  title={Robust stabilization of linear systems with norm-bounded time-varying uncertainty},
  author={Zhou, Kemin and Khargonekar, Pramod P},
  journal=SCL,
  volume={10},
  number={1},
  pages={17--20},
  year={1988},
  publisher={Elsevier}
}

@article{paszke2019pytorch,
  title={Pytorch: An imperative style, high-performance deep learning library},
  author={Paszke, Adam and Gross, Sam and Massa, Francisco and Lerer, Adam and Bradbury, James and Chanan, Gregory and Killeen, Trevor and Lin, Zeming and Gimelshein, Natalia and Antiga, Luca and others},
  journal={Advances in neural information processing systems},
  volume={32},
  year={2019}
}

@article{willems2005note,
  title={A note on persistency of excitation},
  author={Willems, Jan C and Rapisarda, Paolo and Markovsky, Ivan and De Moor, Bart LM},
  journal={Systems \& Control Letters},
  volume={54},
  number={4},
  pages={325--329},
  year={2005},
  publisher={Elsevier}
}

@article{codevilla2018end,
  title={End-to-end driving via conditional imitation learning},
  author={Codevilla, Felipe and M{\"u}ller, Matthias and L{\'o}pez, Antonio and Koltun, Vladlen and Dosovitskiy, Alexey},
  journal=ICRA,
  pages={4693--4700},
  year={2018},
  organization={IEEE}
}

@article{sultangazin2021exploiting,
  title={Exploiting the experts: Learning to control unknown SISO feedback linearizable systems from expert demonstrations},
  author={Sultangazin, Alimzhan and Fraile, Lucas and Tabuada, Paulo},
  journal=CDC21,
  pages={5789--5794},
  year={2021},
  organization={IEEE}
}

@article{east2022imitation,
  title={Imitation learning of stabilizing policies for nonlinear systems},
  author={East, Sebastian},
  journal=ECC,
  volume={68},
  pages={100678},
  year={2022},
  publisher={Elsevier}
}

@article{tu2022sample,
  title={On the sample complexity of stability constrained imitation learning},
  author={Tu, Stephen and Robey, Alexander and Zhang, Tingnan and Matni, Nikolai},
  journal=L4DC4,
  pages={180--191},
  year={2022},
  organization={PMLR}
}

@article{joshi2002design,
  title={Design of norm-bounded and sector-bounded LQG controllers for uncertain systems},
  author={Joshi, SM and Kelkar, AG},
  journal=JOTA,
  volume={113},
  pages={269--282},
  year={2002},
  publisher={Springer}
}

@article{Berberich2020b,
  author={Berberich, Julian and Allgöwer, Frank},
  journal=ECC, 
  title={A trajectory-based framework for data-driven system analysis and control}, 
  year={2020},
  volume={},
  number={},
  pages={1365-1370}}

@article{bertsekas1997nonlinear,
  title={Nonlinear programming},
  author={Bertsekas, Dimitri P},
  journal=JORS,
  volume={48},
  number={3},
  pages={334--334},
  year={1997},
  publisher={Taylor \& Francis}
}

@article{zhu2022development,
  title={Development of Support Programs for Solving BMI Problems by Overbounding Approximation Method},
  author={Zhu, Jiulin and Sebe, Noboru},
  journal={SICE International Symposium on Control Systems},
  pages={9--15},
  year={2022},
  organization={IEEE}
}

@book{nesterov2003introductory,
  title={Introductory lectures on convex optimization: A basic course},
  author={Nesterov, Yurii},
  volume={87},
  year={2003},
  publisher={Springer Science \& Business Media}
}

@article{ghadimi2016mini,
  title={Mini-batch stochastic approximation methods for nonconvex stochastic composite optimization},
  author={Ghadimi, Saeed and Lan, Guanghui and Zhang, Hongchao},
  journal=MP,
  volume={155},
  number={1-2},
  pages={267--305},
  year={2016},
  publisher={Springer}
}
